\documentclass[reqno,12pt,letterpaper]{amsart}
\usepackage{amsmath,amssymb,amsthm,graphicx,mathrsfs,url,bbm}
\usepackage[usenames,dvipsnames]{color}
\usepackage[colorlinks=true,linkcolor=Red,citecolor=Green]{hyperref}
\usepackage{amsxtra}

\def\?[#1]{\textbf{[#1]}\marginpar{\Large{\textbf{??}}}}

\let\epsilon=\varepsilon 

\newcommand{\DD}{{\mathcal D}}
\newcommand{\HH}{{\mathcal H}}
\newcommand{\OO}{{\mathcal O}}
\newcommand{\RR}{{\mathbb R}}
\newcommand{\ZZ}{{\mathbb Z}}

\newcommand{\UU}{{\mathcal U}}
\newcommand{\NN}{{\mathcal N}}
\newcommand{\CC}{{\mathbb C}}
\newcommand{\cchi}{\Tilde{\chi}}
\newcommand{\CI}{{{\mathcal C}^\infty}}
\newcommand{\CIb}{{{\mathcal C}^\infty_{b}}}
\newcommand{\CIc}{{{\mathcal C}^\infty_{\rm{c}}}}
\newcommand{\In}{\textrm{in}}
\newcommand{\out}{\textrm{out}}

\let\sharp=\OO

\newtheorem{thm}{Theorem}

\newtheorem{lem}{Lemma}

\numberwithin{equation}{section}
\numberwithin{lem}{section}

\DeclareMathOperator{\Det}{det}
\DeclareMathOperator{\Diff}{Diff}

\DeclareMathOperator{\dist}{dist}

\DeclareMathOperator{\Spec}{Spec}
\DeclareMathOperator{\comp}{comp}

\DeclareMathOperator{\id}{id}

\DeclareMathOperator{\Ran}{Ran}
\DeclareMathOperator{\rank}{rank}

\let\Re=\Real

\DeclareMathOperator{\supp}{supp}

\DeclareMathOperator{\tr}{tr}

\title{Resonances as Viscosity Limits for Black Box  Perturbations}
\author{Haoren Xiong}
\email{xiong@math.berkeley.edu}
\address{Department of Mathematics, University of California,
Berkeley, CA 94720, USA}

\begin{document}

\begin{abstract}
We show that the complex absorbing potential (CAP) method for computing
scattering resonances applies to an abstractly defined class of black box perturbations of the Laplacian in $\RR^n$ which can be analytically extended from $\RR^n$ to a conic neighborhood in $\CC^n$ near infinity. The black box setting allows a unifying treatment of diverse problems ranging from obstacle scattering to scattering on finite volume surfaces. 
\end{abstract}

\maketitle

\section{Introduction and statement of results}
\label{introduction}

The complex absorbing potential (CAP) method has been used as a computational
tool for finding scattering resonances -- see Riss--Meyer \cite{RiMe} and Seideman--Miller \cite{semi} for an early treatment and Jagau et al \cite{Jag} for some recent developments. Zworski \cite{Zw-vis} showed that scattering resonances of $-\Delta+V$, $V \in L^\infty_{\comp} $, are limits of eigenvalues of $-\Delta+V-i\epsilon x^2$ as $\epsilon\to 0+$ . The situation is very different for potentials of the Wigner--von Neumann type, in which case Kameoka and Nakamura \cite{Kameoka2020} showed that the corresponding limits exist away from a discrete set of thresholds. Using an approach closer to \cite{Kameoka2020} than \cite{Zw-vis}, the author extended Zworski's result to potentials which are exponentially decaying \cite{xiong2020}. In this paper we show that the CAP method is also valid for an abstractly defined class of \textit{black box} perturbations of the Laplacian in $\RR^n$ which can be analytically extended from $\RR^n$ to a conic neighborhood in $\CC^n$ near infinity.

We formulate black box scattering using the abstract setting introduced by Sj\"ostrand and Zworski in \cite{SZ1} except that the operator $P$ is not assumed to be equal to $-\Delta$ near infinity. For that we follow Sj\"ostrand \cite{sj97} and assume that $P$ is a dilation analytic perturbation of $-\Delta$ near infinity. The black box formalism allows an abstract treatment of diverse scattering problems without addressing the details of specific situations -- see Examples 1--3 later in this section. We recall the setup as follows:  

Let $\HH$ be a complex separable Hilbert space with an orthogonal decomposition:
\begin{equation}
\label{eqn:Hilbert space}
    \HH = \HH_{R_0} \oplus L^2 (\RR^n\setminus B(0,R_0)),
\end{equation}
where $B(x,R)=\{ y\in\RR^n : |x-y|<R \}$ and $R_0$ is fixed. The corresponding orthogonal projections will be denoted by $u \mapsto u|_{B(0,R_0)}$, and $u\mapsto u|_{\RR^n\setminus B(0,R_0)}$ or simply by the characteristic function $1_L$ of the corresponding set $L$. We consider an unbounded self-adjoint operator
\begin{equation}
\label{eqn:P}
    P : \HH \to \HH \quad\textrm{with domain }\DD.
\end{equation}
We assume that 
\begin{equation}
\label{eqn:domain D}
    \DD|_{\RR^n\setminus B(0,R_0)} \subset H^2(\RR^n\setminus B(0,R_0)),
\end{equation}
and conversely, $u\in\DD$ if $u\in H^2(\RR^n\setminus B(0,R_0))$ and $u$ vanishes near $B(0,R_0)$; and that
\begin{equation}
\label{eqn:P compact}
    1_{B(0,R_0)} (P + i)^{-1}\textrm{ is compact}.
\end{equation}
We also assume that,
\begin{equation}
\label{eqn:Q defn}
\begin{gathered}
    1_{\RR^n\setminus B(0,R_0)} Pu = Q(u|_{\RR^n\setminus B(0,R_0)}),\quad\textrm{ for all }u\in\DD,  \\
    Q = -\sum_{j,k=1}^n \partial_{x_j}(g^{jk}(x)\partial_{x_k}) + c(x),\quad g^{jk},c\in \CIb(\RR^n).
\end{gathered}
\end{equation}
Here $\CIb$ denotes the space of $\CI$ functions with all derivatives bounded. Note that if $\psi\in\CIb(\RR^n)$ is constant near $B(0,R_0)$, then there is a natural way to define the multiplication: $\HH\owns u \mapsto \psi u \in\HH$, and we have $\psi u\in\DD$ if $u\in\DD$.

It is further assumed that $Q$ is formally self-adjoint, i.e. $g^{jk},\,c$ are real-valued functions on $\RR^n$ satisfying
\begin{equation}
\label{eqn:Qelliptic}
\begin{gathered}
    |\sum_{j,k=1}^n g^{jk}(x)\xi_j\xi_k| \geq C^{-1}|\xi|^2, \\
    \sum_{j,k=1}^n g^{jk}(x)\xi_j\xi_k + c(x) \to \xi^2,\ |x|\to\infty.
\end{gathered}
\end{equation}

We will use the method of complex scaling -- see \S \ref{complex scaling} to define the resonances of $P$. For that we follow \cite{sj97} to make the  following assumptions:
\begin{equation}
\label{Analytic extension}
\begin{gathered}
    \textrm{There exist }\theta_0\in [0,\pi/8],\,\delta>0,\textrm{ and }R \geq R_0,\textrm{ such that}\\
    \textrm{the coefficients } g^{jk}(x), c(x)\textrm{ of }Q \textrm{ extend analytically in }x\textrm{ to } \\
    \{ s\omega : \omega\in\CC^n,\, \dist(\omega,\mathbb{S}^{n-1})<\delta,\ s\in\CC,\, |s|>R,\, \arg s\in(-\delta,\theta_0 + \delta) \} \\
    \textrm{ and the second half of \eqref{eqn:Qelliptic} remains valid in this larger set.} 
\end{gathered}
\end{equation}
We can now define the resonances $z_j$ of $P$ in $\{z\in \CC\setminus\{0\} : \arg z > -2\theta_0\}$ as the eigenvalues of $P$ on a suitable contour in $\CC^n$, see \cite{SZ1} and \S \ref{complex scaling}.

We now introduce a \textit{regularized} operator,
\begin{equation}
\label{eqn:Peps}
    P_\epsilon := P - i\epsilon(1-\chi(x))x^2,\quad\epsilon>0,
\end{equation}
where $\chi\in\CIc(\RR^n)$ is equal to $1$ near $\overline{B(0,R_0)}$; $ x^2 := x_1^2 + \cdots + x_n^2 $. It follows from \S \ref{section:Peps} that $ P_\epsilon $
is an unbounded operator on $ \HH $ with a discrete spectrum. We have
\begin{thm}
\label{t:1}
Suppose that $ \{ z_j ( \epsilon) \}_{j=1}^\infty $ are
the eigenvalues of $ P_\epsilon $. Then, uniformly on any compact subset of the sector $ \{ z\in\CC\setminus\{0\}\,:\,-2\theta_0 < \arg z < 3\pi/2 + 2\theta_0 \} $,
\[   z_j ( \epsilon ) \to z_j ,  \ \ \epsilon \to 0 + , \]
where $ z_j $ are the resonances of $ P $.
\end{thm}

\noindent
{\bf Remark:} We will prove a more precise version of this theorem in \S\ref{section:poc}: it involves the multiplicities of $z_j$ and $z_j(\epsilon)$ defined in \S\ref{complex scaling} and \S\ref{section:Peps} respectively. The term viscosity is motivated by the viscosity definition of Pollicott--Ruelle resonances given in Dyatlov–Zworski \cite{dyatlov2015}.

Fixed complex absorbing potentials have already been used in mathematical literature on scattering resonances. Stefanov \cite{stefanov2005} showed that semiclassical resonances
close to the real axis can be well approximated using eigenvalues of the Hamiltonian modified by a complex absorbing potential. For applications of fixed complex absorbing potentials in generalized geometric settings see for instance Nonnenmacher--Zworski \cite{NZ1}, \cite{NZ2} and Vasy \cite{V}. The analogous results to Theorem \ref{t:1} were proved for Pollicott--Ruelle resonances in \cite{dyatlov2015}, for kinetic Brownian motion by Drouot \cite{Drouot}, for gradient flows by Dang--Rivi\`{e}re \cite{dang2017} (following earlier work of Frenkel--Losev--Nekrasov \cite{Frenkel}), and for $0$th order pseudodifferential operators, motivated by problems in fluid mechanics, by Galkowski--Zworski \cite{galkowski2019}. 

\medskip
\noindent
{\bf Example 1. Obstacle scattering.} Suppose that $\OO\subset\overline{B(0,R_0)}$ is an open set such that $\partial\OO$ is a smooth hypersurface in $\RR^n$. Let $\HH = L^2(\RR^n\setminus\OO)$, and $P=-\Delta|_{\RR^n\setminus\OO}$ on the exterior domain realized with any self-adjoint boundary
conditions on $\partial\OO$. For instance, the Dirichlet boundary condition
\[
    \DD = \{ u\in H^2(\RR^n\setminus\OO) : u|_{\partial\OO} = 0 \}
\]
or the Neumann/Robin boundary condition
\[
    \DD = \{ u\in H^2(\RR^n\setminus\OO) : \partial_\nu u + \eta u|_{\partial\OO} = 0 \}
\]
where $\partial_\nu$ is the normal derivative with respect to $\partial\OO$ and $\eta$ is a real-valued smooth function on $\partial\OO$. Theorem \ref{t:1} shows that the eigenvalues of $P-i\epsilon x^2$ converge to the resonances of $P$ (the irrelevance of the missing $i\epsilon\chi (x) x^2$ term comes from continuity of resonances under compactly supported perturbations -- see Stefanov \cite{stefanov1994}).

\medskip
\noindent
{\bf Example 2. Scattering on asymptotically Euclidean space.} Let $M$ be a real analytic manifold which is diffeomorphic to $\RR^n$ near infinity and equipped with a real analytic metric $g$ which is asymptotically Euclidean. More precisely, let $g_{ij}=\delta_{ij} + h_{ij}$ be the metric tensor then we assume that $h_{ij}(x)$ extend analytically in $x$ to 
\[
    \{ s\omega : \omega\in\CC^n,\, \dist(\omega,\mathbb{S}^{n-1})<\delta,\ s\in\CC,\, |s|>R,\, \arg s\in(-\delta,\theta_0 + \delta) \}
\]
for some $\theta_0\in [0,\pi/8]$, $\delta>0$, $R\geq R_0$, and that $h_{ij} \to 0$ in this larger set. We put $P=-\Delta_g$, the Laplace--Beltrami operator with respect to the metric $g$, then all the black box assumptions are satisfied. Suppose that $\chi\in\CIc(M;[0,1])$ is equal to $1$ near some compact set $K$ and that $M\setminus K$ is diffeomorphic to $\RR^n\setminus \overline{B(0,R_0)}$. Then the operator $-\Delta_g-i\epsilon(1-\chi(x))x^2$ has a discrete spectrum for $\epsilon>0$ and the eigenvalues converge to the resonances of $-\Delta_g$ uniformly on compact subsets of $-2\theta_0 < \arg z < 3\pi/2 + 2\theta_0$.

\medskip
\noindent
{\bf Example 3. Scattering on finite volume surfaces.} This example was already discussed in \cite{Zw-vis} but this paper provides a complete proof via the black box setting. Consider the modular surface $M = SL_2(\ZZ)\backslash \mathbb{H}^2$ (or any surfaces with cusps -- see \cite[\S 4.1, Example 3]{res}) equipped with the Poincar\'e metric $g$ and $ \Delta_M \leq 0 $
the Laplacian on $ M $. We choose the fundamental domain of $SL_2(\ZZ)$ to be $\{ x+iy\in\mathbb{H}^2 : |x|\leq 1/2, x^2 + y^2\geq 1  \}$ then $\Delta_M$ 
in the cusp $ y > 1 $ is given by $ y^{2}( \partial_x^2 + \partial_y^2 ) $. Let $r=\log y$, $\theta =2\pi x$, then $M$ in $(r,\theta)$ coordinates admits the following decomposition: 
\[
    M = M_0\cup M_1,\  (M_1,g|_{M_1})= ([0,\infty)_r\times\mathbb{S}_\theta^1,\,dr^2 + (2\pi)^{-2} e^{-2r}{d\theta} ^2),\ \mathbb{S}^1 = \RR/2\pi\ZZ.  
\]
We recall the black box setup in this case from \cite[\S 4.1, Example 3]{res}. Let
\[
    \HH = \HH_0\oplus L^2([0,\infty), dr),\quad \HH_0 = L^2(M_0)\oplus \HH_0^0,
\]
where (with $\ZZ^* :=\ZZ\setminus\{0\}$)
\[
    \HH_0^0 = \left\{ \{a_n(r)\}_{n\in\ZZ^*} : a_n\in L^2([0,\infty)),\ \sum_{n\in\ZZ^*}\int_0^\infty |a_n(r)|^2 dr < \infty\right\}.
\]
We can identify $L^2(M)$ with $\HH$ via the following isomorphism:
\[
\begin{gathered}
    \iota : L^2(M) \owns u \mapsto \big( u|_{M_0}, \{e^{-r/2}u_n(r)\}_{n\in\ZZ^*}, e^{-r/2} u_0(r) \big) \in \HH, \\
    u_n(r) := \frac{1}{2\pi} \int_{\mathbb{S}^1} u(r,\theta) e^{-in\theta} d\theta,\quad r>0.
\end{gathered}
\]
Then $P:= -\Delta_M -1/4$ is a black box Hamiltonian on $\HH$ which equals $-\partial_r^2$ on $L^2([0,\infty), dr)$ -- see \cite[\S 4.1, Example 3]{res}. In the language of Theorem \ref{t:1} and in $ ( x, y ) $ coordinates
\[
    P_\epsilon = - \Delta_M - 1/4 - i \epsilon ( 1 - \chi ( y ) ) (\log y )^2 \Pi_0,\quad \Pi_0 u(x,y) := \int_{-1/2}^{1/2} u(x',y)\,dx'.
\]
where $ \chi \in \CIc ( [0, \infty ) ) $, $ \chi ( y ) \equiv 1 $ for $ y < 2$ and $ \chi (y ) \equiv 0 $ for $ y > 3 $.
The eigenvalues of $P_\epsilon$ converge to the resonances of $P$ uniformly on compact subsets of $ \arg z > - \pi /4 $. Equivalently if we define $s ( \epsilon ) \in \Sigma_\epsilon \Leftrightarrow s( \epsilon ) ( 1 - 
s ( \epsilon ) ) - 1/4 \in \Spec ( P_\epsilon )$, 
then the limit points of $ \Sigma_\epsilon $, $ \epsilon \to 0+ $, in 
$ \Re s < 1/2 $, $ \arg(s-1/2) \neq 11\pi/8 $ are given by the nontrivial 
zeros of $ \zeta ( 2s ) $ where $\zeta$ is the Riemann zeta function -- see \cite[Example 2]{Zw-vis} and \cite[\S 4.4 Example 3]{res}.

The paper is organized as follows. In \S\ref{complex scaling} we review the method of complex scaling and define the resonances of $P$ as the eigenvalues of the complex scaled operator $P_\theta$. In \S\ref{section:Peps} we show that $P_\epsilon$ has a discrete spectrum in $\CC\setminus e^{-i\pi/4}[0,\infty)$, which is invariant under complex scaling. Since our operator is an abstract perturbation of $-\Delta$, in \S\ref{section:N operator} we use a different method from \cite{Zw-vis} and \cite{xiong2020} to characterize the eigenvalues of $P_{\epsilon,\theta}$, $\epsilon\geq 0$. More precisely, we use a reference operator reviewed in \S\ref{reference operator} to introduce the Dirichlet-to-Neumann operator ${\NN}_{\epsilon,\theta}(z)$ associated with $P_{\epsilon,\theta}$ and an artificial smooth obstacle $\OO$. The artificial obstacle problem is needed to separate the abstract black box from the differential operator outside. The operator $\NN_{\epsilon,\theta}(z)$ is well-defined for all $z$ except for a discrete set depending on the obstacle, and we show that the eigenvalues of $P_{\epsilon,\theta}$ can be identified with the poles of $z\mapsto \NN_{\epsilon,\theta}(z)^{-1}$, with agreement of multiplicities. In \S\ref{section:deform obstacle} we show that the obstacle can be chosen so that the corresponding $\NN_{\epsilon,\theta}(z)$ is well-defined near the resonances $z_j$. The proof of Theorem \ref{t:1} is completed in \S \ref{section:poc} by obtaining further estimates on $\NN_{\epsilon,\theta}(z)$.
 
\medskip
\noindent
{\bf Notation.}
We use the following notation: $ f =  O_\ell (
 g   )_H $ means that
$ \|f \|_H  \leq C_\ell  g $ where the norm (or any seminorm) is in the
space $ H$, and the constant $ C_\ell  $ depends on $ \ell $. When either
$ \ell $ or
$ H $ are absent then the constant is universal or the estimate is
scalar, respectively. When $ G = O_\ell ( g ) : {H_1\to H_2 } $ then
the operator $ G : H_1  \to H_2 $ has its norm bounded by $ C_\ell g $.
Also when no confusion is likely to result, we denote the operator
$ f \mapsto g f $ where $ g $ is a function by $ g $.

\medskip
\noindent
{\sc Acknowledgments.} The author would like to thank
Maciej Zworski for helpful discussions. This project was supported in part by
the National Science Foundation grant DMS-1901462.

\section{Preliminaries}
\label{preliminaries}

\subsection{Review of Complex Scaling}
\label{complex scaling}
Complex scaling has been a standard technique in resonance theory since the works of Aguilar--Combes \cite{AgCo}, Balslev--Combes \cite{BaCo} and Simon \cite{simon}. Here we follow rather closely the presentation in \cite{sj97} since our assumptions on the operator $P$ is weaker than \cite{SZ1}.

A smooth submanifold $\Gamma\subset\CC^n$ is said to be totally real if $T_x\Gamma\cap i T_x\Gamma = \{0\}$ for every $x\in\Gamma$, where we identify $T_x\Gamma$ with a real subspace of $T_x\CC^n\simeq \CC^n$. We say that $\Gamma$ is maximally totally real if $\Gamma$ is totally real and of maximal (real) dimension $n$, the natural example is $\Gamma=\RR^n$. Let $\Gamma\subset\CC^n$ be smooth and of real dimension $n$, then locally $\Gamma$ can be represented using real coordinates: $\RR^n \owns x \mapsto f(x)\in\Gamma$. Let $\Tilde{f}$ be an almost analytic extension of $f$ so that  $\Bar{\partial}\Tilde{f}$ vanishes to infinite order on $\RR^n$. Let $x\in\RR^n$, then since $d\Tilde{f}(x)$ is complex linear, $i T_{f(x)}\Gamma = d\Tilde{f}(x)(i T_x \RR^n)$. Hence $\Gamma$ is totally real in a neighborhood of $f(x)$ if and only if $d\Tilde{f}(x)$ is injective, i.e. $\Det d f(x)\neq 0$.

Let $\Omega\subset\CC^n$ be an open neighborhood of $\Gamma$ such that $\Gamma$ is closed in $\Omega$, and let 
\[
    A(z,D_z) = \sum_{|\alpha|\leq m} a_\alpha (z) D_z^\alpha,\quad D_{z_j}:=\frac{1}{i}\partial_{z_j},\quad D_z^\alpha = D_{z_1}^{\alpha_1}\cdots D_{z_n}^{\alpha_n},
\]
be a differential operator on $\Omega$ with holomorphic coefficients. Define $A_\Gamma:\CI(\Gamma)\to \CI(\Gamma)$ by
\begin{equation}
\label{eqn:A_Gamma}
    A_\Gamma u = (A\Tilde{u})|_{\Gamma},
\end{equation}
where $\Tilde{u}$ is an almost analytic extension of $u$, that is, a smooth extension of $u$ to a neighborhood of $\Gamma$ such that $\Bar{\partial}\Tilde{u}$ vanishes to infinite order on $\Gamma$. $A_\Gamma$ is then a differential operator on $\Gamma$ with smooth coefficients, and for the principal symbols we have 
\[    a_\Gamma = a|_{T^* \Gamma},
\]
where $a$ is the principal symbol of $A$. 

We recall a deformation result from \cite[Lemma 3.1]{SZ1}:
\begin{lem}
\label{lem:deformation}
Suppose that $W\subset\RR^n$ is open and that $F:[0,1]\times W \owns (s,x) \mapsto F(s,x)\in \CC^n$, is a smooth proper map satisfying for all $s\in[0,1]$
\[    \Det\partial_x F(s,x)\neq 0,\quad\textrm{and }\, x\mapsto F(s,x)\textrm{ is injective,}
\]
and assume that $x\in W\setminus K \implies F(s,x)=F(0,x)$  for some compact $K\subset W$.

Let $A(z,D_z)$ be a differential operator with holomorphic coefficients defined in a neighborhood of $F([0,1]\times W)$ such that for $0\leq s\leq 1$ and $\Gamma_s := F(\{s\}\times W)$,  $A_{\Gamma_s}$ is elliptic. 

If $u_0\in\CI(\Gamma_0)$ and $A_{\Gamma_0} u_0$ extends to a holomorphic function in a neighborhood of $F([0,1]\times W)$, then the same holds for $u_0$.
\end{lem}

The lemma will be applied to a family of deformations of $\RR^n$ in $\CC^n$. We aim to restrict the operators $P_\epsilon,\ \epsilon\geq 0$, to the corresponding totally real submanifolds. For given $\alpha_0 > 0$ and $R_1>R_0$, we can construct a smooth function 
\[ [0,\theta_0]\times [0,\infty)\owns (\theta,t) \mapsto g_\theta(t)\in\CC, \] 
injective for every $\theta$, with the following properties: 
\begin{enumerate}
    \item $g_\theta(t)=t$ for $0\leq t \leq R_1$,
    \item $0\leq \arg g_\theta(t)\leq \theta,\quad \partial_t g_\theta(t) \neq 0$,
    \item $\arg g_\theta(t) \leq \arg\partial_t g_\theta(t) \leq \arg g_\theta(t) + \alpha_0$,
    \item $g_\theta(t)=e^{i\theta} t$ for $t\geq T_0$, where $T_0$ depends only on $\alpha_0$ and $R_1$.
\end{enumerate}
We now define the totally real submanifolds, $\Gamma_\theta$, as images of $\RR^n$ under the maps
\[    f_\theta :\RR^n\owns x=t\omega \mapsto g_\theta(t)\omega\in\CC^n,\ t=|x|.
\]
Then a dilated operator $P_\theta$ can be defined as follows. Let 
\[    \HH_\theta = \HH_{R_0} \oplus L^2(\Gamma_\theta\setminus B(0,R_0)),
\]
where $B(0,R_0)$ denotes the real ball as before. If $\chi\in \CIc(B(0,R_1))$ is equal to $1$ near $\overline{B(0,R_0)}$, we put
\[    \DD_\theta = \{ u\in\HH_\theta : \chi u\in\DD,\ (1-\chi)u\in H^2(\Gamma_\theta\setminus B(0,R_0)) \}.
\]
Let $P_\theta$ be the unbounded operator $\HH_\theta \to \HH_\theta$ with domain $\DD_\theta$, given by 
\[ P_\theta u := P(\chi u) + Q_{\theta} ((1-\chi)u),\quad Q_\theta := -\sum_{j,k=1}^n (\partial_{z_j}(g^{jk}(z)\partial_{z_k}) + c(z))|_{\Gamma_\theta}. \]
These definitions do not depend on the choice of $\chi$. 

We recall some properties of the dilated Laplacian from \cite[\S 3]{SZ1}. Let
\[
    \Delta_\theta := (\Delta_z)|_{\Gamma_\theta},\quad x_\theta := z|_{\Gamma_\theta}.
\]
Parametrizing $\Gamma_\theta$ by $[0,\infty)\times\mathbb{S}^{n-1}\owns (t,\omega)\mapsto g_\theta(t)\omega$, we obtain 
\begin{equation}
\label{eqn:Delta_Gamma_theta}
-\Delta_\theta = (g_\theta '(t)^{-1}D_t)^2 - i(n-1)(g_\theta(t)g_\theta'(t))^{-1} D_t +g_\theta(t)^{-2} D_\omega^2,
\end{equation}
where $D_t=-i\partial_t$ and $D_\omega^2 = -\Delta_{\mathbb{S}^{n-1}}$. If ${\omega^*} ^2$ denotes the principal symbol of $D_\omega^2$ and we let $\tau$ be the dual variable of $t$, then the principal symbol of $-\Delta_{\theta}$ is 
\[ \sigma(-\Delta_{\theta}) = g_\theta'(t)^{-2} \tau^2 + g_\theta(t)^{-2} {\omega^*}^2, \]
so pointwise on $\Gamma_\theta$, $-\Delta_{\theta}$ is elliptic and the principal symbol takes values in an angle of size $\leq 2\alpha_0$, while globally, $\sigma(-\Delta_{\theta})$ takes values in the sector $-2\theta-2\alpha_0 \leq \arg z\leq 0$. The basic result based on ellipticity at infinity is 
\begin{equation}
\label{eqn:Delta_theta resolv bound}
\begin{gathered}
    -2\theta+\delta < \arg z < 2\pi - 2\theta - \delta,\quad |z|>\delta \implies \\
    (-\Delta_\theta - z)^{-1} = O_\delta (|z|^{\frac{j-2}{2}}): L^2(\Gamma_\theta)\to H^j(\Gamma_\theta),\ j=0,1,2.
\end{gathered}
\end{equation}
This follows from \cite[Lemmas 3.2--3.5 and \S4]{SZ1} applied with $P=-\Delta$.

$P_\theta$, as a perturbation of $-\Delta_\theta$, is also elliptic -- see \cite[\S 5]{sj97}. More precisely, choosing $R_1$ large enough, it follows from the assumptions \eqref{eqn:Qelliptic} and \eqref{Analytic extension} that
\begin{equation}
\label{P_theta elliptic}
\begin{gathered}
    \textrm{In $\Gamma_\theta\setminus B(0,R_0)$, $P_\theta$ is an elliptic differential operator whose principal} \\
\textrm{symbol pointwise on $\Gamma_\theta$ takes its values in an angle of size $\leq 3\alpha_0$}, \\
\textrm{and globally in a sector $-2\theta-3\alpha_0 \leq \arg z\leq \alpha_0$}.
\end{gathered}
\end{equation}

\begin{equation}
\label{P_theta at infinity}
\begin{gathered}
    \textrm{The coefficients of $P_\theta - e^{-2i\theta}(-\Delta)$ tend to zero when $\Gamma_\theta\owns x\to\infty$}, \\
    \textrm{where we identify $\Gamma_\theta$ and $\RR^n$, by means of $f_\theta$.}
\end{gathered}
\end{equation}

We recall some basic results about $P_\theta$ from \cite[\S 5]{sj97}:
\begin{lem}
\label{lem:Ptheta Fredholm}
If $z\in\CC\setminus\{0\}$, $\arg z\neq -2\theta$, then $P_\theta - z: \DD_\theta \to \HH_\theta$ is a Fredholm operator of index $0$. In particular the spectrum of $P_\theta$ in $\CC\setminus e^{-2i\theta}[0,\infty)$ is discrete.
\end{lem}

\begin{proof}
We follow closely the proof of \cite[Lemma 3.2]{SZ1} (see also \cite[Theorem 4.36]{res}) except that $P_\theta$ is more general here. We shall invert $P_\theta - z$ modulo compact operators. On the complex contour $\Gamma_\theta$ we introduce a smooth partition of unity: $1 = \chi_1 + \chi_2 + \chi_3$ with $\supp\chi_1 \subset B(0,R_1)$, $\supp\chi_3$ contained in the region where $\Gamma_\theta\owns x_\theta=e^{i\theta}x,\ x\in\RR^n$, $\supp\chi_2$ compact and disjoint from $\overline{B(0,R_0)}$. Let $\cchi_j$ have the same properties as the $\chi_j$ except that they do not form a partition of unity, satisfying $\cchi_j=1$ near $\supp\chi_j$. Now we put
\begin{equation}
\label{eqn:E(z)}
    E(z) = \cchi_1(P-z_0)^{-1}\chi_1 + S(z)\chi_2 + \cchi_3\, e^{2i\theta}(-\Delta-e^{2i\theta}z)^{-1}\chi_3,
\end{equation}    
where $z_0\in\CC\setminus\RR$ and $S(z)$ is a properly supported parametrix of the elliptic operator $P_{\theta} - z$ in $\Gamma_\theta\setminus B(0,R_0)$. Then we have
\begin{equation}
\label{eqn:right appro inverse}
    (P_{\theta} - z) E(z) = I + K(z) + K_1(z),
\end{equation}
where
\[
\begin{gathered}
    K(z) = (z_0 - z)\cchi_1(P-z_0)^{-1}\chi_1 + [P,\cchi_1](P-z_0)^{-1} \chi_1 + ((P_{\theta} - z)S(z) - I)\chi_2\\
    \qquad  + [-e^{-2i\theta}\Delta,\cchi_3]\,e^{2i\theta}(-\Delta-e^{2i\theta}z)^{-1}\chi_3, \\
    K_1(z) = (P_\theta - (-e^{-2i\theta}\Delta))\cchi_3\, e^{2i\theta}(-\Delta-e^{2i\theta}z)^{-1}\chi_3.
\end{gathered}
\]
Using \eqref{P_theta at infinity} we may assume that $\supp\chi_3\subset\{z\in\CC^n:|z|\geq T\}$ for $T$ sufficiently large such that $\|K_1(z)\|_{\HH_\theta\to\HH_\theta}\leq 1/2$, thus $I+K_1(z)$ is invertible and we get
\[    (P_{\theta} - z)E(z)(I+K_1(z))^{-1} = I + K(z)(I+K_1(z))^{-1}.
\]
It follows from the assumptions \eqref{eqn:domain D} and \eqref{eqn:P compact} that $K(z)$ is compact: $\HH_\theta\to \HH_\theta$, thus $E(z)(I+K_1(z))^{-1}$ is an approximate right inverse. The construction of an approximate left inverse is similar, we omit the details and refer to \cite[Lemma 3.2]{SZ1}. 

It remains to show that $P_\theta - z$ is invertible for some $z\in\CC\setminus e^{-2i\theta}[0,\infty)$. For $z_0 = i L,\ L\geq 1$, we can replace \eqref{eqn:E(z)} with
\begin{equation}
\label{eqn:E(z0)}
    E(z_0) = \cchi_1 (P-z_0)^{-1} \chi_1 + (1-\chi_0) (-\Delta_\theta - z_0)^{-1} (1-\chi_1),
\end{equation}
where $\chi_1\in\CIc(B(0,R_1))$ is equal to $1$ near $\supp\chi_0$ and $\chi_0 = 1$ on $B(0,R_1 - \delta)$, for some $\delta>0$ small. Then \eqref{eqn:right appro inverse} still holds with $K(z_0)$, $K_1(z_0)$ given by
\[
\begin{gathered}
    K(z_0) = [P,\cchi_1](P-z_0)^{-1}\chi_1 + [\Delta_\theta,\chi_0](-\Delta_\theta - z_0)^{-1} (1-\chi_1), \\
    K_1(z_0) = (P_\theta - (-\Delta_\theta))(1-\chi_0)(-\Delta_\theta - z_0)^{-1} (1-\chi_1).
\end{gathered}
\]
Choosing $R_1$ sufficiently large, we may assume by \eqref{eqn:Delta_theta resolv bound} and \eqref{P_theta at infinity} that
$\|K_1(z_0)\|_{\HH_\theta\to\HH_\theta} \leq 1/2$, for all $z_0=iL$,  $L\geq 1$. Then we get
\[ (P_\theta - z_0)E(z_0)(I+K_1(z_0))^{-1} = I + K(z_0)(I+K_1(z_0))^{-1}. \]
It follows from \eqref{eqn:Delta_theta resolv bound} that $K(iL)=O(L^{-1/2}):\HH_\theta\to\HH_\theta$, thus $P_\theta - z_0$ is inveritible for $z_0=iL,\ L\gg 1$ and we have 
\begin{equation}
\label{eqn:Ptheta z0inverse}
    (P_\theta-z_0)^{-1} = E(z_0)(I+K_1(z_0))^{-1} (I + K(z_0)(I+K_1(z_0))^{-1})^{-1} ,
\end{equation}
which completes the proof.
\end{proof}

\begin{lem}
\label{lem:dim_ker Ptheta}
Assume that $0\leq \theta_1 < \theta_2\leq \theta_0$ and let $z_0\in\CC\setminus e^{-2i[\theta_1,\theta_2]}[0,\infty)$. Then
\[    \dim\ker (P_{\theta_1} - z_0) = \dim\ker (P_{\theta_2} - z_0).
\]
\end{lem}

This is identical to \cite[Lemma 3.4]{SZ1} and the proof is the same as there using Lemma \ref{lem:deformation}.

Lemma \ref{lem:dim_ker Ptheta} shows that the spectrum in the sector  $-2\theta_0<\arg z\leq 0$ is independent of $\theta$ in the following sense: We say that $z\in \CC\setminus\{0\}$, $-2\theta_0<\arg z\leq 0$ is a resonance for $P$ if and only if $z\in\Spec(P_\theta)$ with $-2\theta <\arg z\leq 0$ for some $\theta\in (0,\theta_0]$. For such a resonance $z_0\in e^{-2i[0,\theta)}(0,\infty)$, the spectral projection 
\begin{equation}
\label{eqn:SpecProj Ptheta}
    \Pi_\theta(z_0) = \frac{1}{2\pi i} \oint_{z_0} (z-P_\theta)^{-1}dz,
\end{equation}
where the integral is over a positively oriented circle enclosing $z_0$ and containing no resonances other than $z_0$, is of finite rank. The restriction of $P_\theta - z_0$ to $\Ran \Pi_\theta(z_0)$ is nilpotent. If $\Tilde{\theta}\in [0,\theta_0]$ is a second number with $z_0\in e^{-2i[0,\Tilde{\theta})}(0,\infty)$, then since Lemma \ref{lem:dim_ker Ptheta} can be extended to $\dim \ker(P_{\theta}-z_0)^k = \dim \ker(P_{\Tilde{\theta}}-z_0)^k$ for all $k$, $\Pi_\theta(z_0)$ and $\Pi_{\Tilde{\theta}}(z_0)$ have the same rank, which by definition is the multiplicity of the resonance $z_0$:
\begin{equation}
\label{eqn:multiplicity z0}
    m(z_0):=\rank \Pi_\theta(z_0),\quad -2\theta<\arg z_0\leq 0.
\end{equation}

\subsection{A reference operator}
\label{reference operator}
As explained in \S \ref{introduction}, to separate the abstract black box from the differential operator outside we introduce a \textit{reference operator} $P^\sharp$ associated with an open set $\OO\subset\RR^n$ containing $\overline{B(0,R_0)}$. We assume that $\partial \OO$ is a smooth hypersurface in $\RR^n$. In the notation of \eqref{eqn:Hilbert space}, we put
\begin{equation}
\label{eqn:H sharp}
    \HH^\sharp := \HH_{R_0}\oplus L^2(\OO\setminus B(0,R_0)).
\end{equation}
The corresponding orthogonal projections are denoted by 
\[
    u\mapsto 1_{B(0,R_0)}\,u = u|_{B(0,R_0)},\quad u\mapsto 1_{\OO\setminus B(0,R_0)}\,u = u|_{\OO\setminus B(0,R_0)}.
\]
If $P$ is a black box Hamiltonian introduced in \S\ref{introduction} with domain $\DD$, then we define
\begin{equation}
\label{eqn:D sharp}
    \begin{gathered}
        \DD^\sharp := \{u\in\HH^\sharp : \psi\in \CIc(\OO),\ \psi=1\textrm{ near }\overline{B(0,R_0)} \Rightarrow \\
        \quad\psi u\in\DD,\ (1-\psi)u\in H^2(\OO)\cap H_0^1(\OO)\}
    \end{gathered}
\end{equation}
and, for any $\psi$ with the property \eqref{eqn:D sharp},
\begin{equation}
\label{eqn:P sharp}
    \begin{gathered}
        P^\sharp : \DD^\sharp \to \HH^\sharp, \\
        P^\sharp u := P(\psi u) + Q((1-\psi)u).
    \end{gathered}
\end{equation}
Assumptions \eqref{eqn:domain D}, \eqref{eqn:Q defn} show that this definition is independent of the choice of $\psi$.

We recall some basic properties of the reference operator from \cite[\S 7]{SZ1}:
\begin{lem}
\label{lem:reference operator}
Suppose that $\OO\subset\RR^n$ is an open set containing $\overline{B(0,R_0)}$ such that $\partial\OO$ is a smooth hypersurface in $\RR^n$. Let $P^\OO$ be the reference operator defined in \eqref{eqn:P sharp}. Then, with $\HH^\sharp$ given by \eqref{eqn:H sharp},
\[
    P^\sharp : \HH^\sharp \to \HH^\sharp,
\]
is a self-adjoint operator with domain $\DD^\sharp$ defined in \eqref{eqn:D sharp}. Furthermore, the resolvent $(P^\sharp + i)^{-1}$ is compact and thus $P^\sharp$ has discrete spectrum which is contained in $\RR$.
\end{lem}
For the proof we refer to Dyatlov--Zworski \cite[Lemma 4.11]{res} and we remark that the arguments there is still valid if we replace the assumption there: $P=-\Delta$ in $\RR^n\setminus B(0,R_0)$, by the assumption \eqref{eqn:Q defn}.

\section{The regularized operator}
\label{section:Peps}

In this section we show that the spectrum of $P_\epsilon$ is invariant under complex scaling. Choosing $R_1$ such that $\supp\chi\subset B(0,R_1)$ when we construct the complex contours $\Gamma_\theta$, the complex absorbing potential $-i\epsilon(1-\chi(x))x^2$ can be analytically extended to $\Gamma_\theta$, thus it defines a multiplication on the following subspace of $\HH_\theta$:
\[ \widehat{\HH}_\theta := \HH_{R_0}\oplus |x_\theta|^{-2} L^2(\Gamma_\theta\setminus B(0,R_0)),\]
where $x_\theta := f_\theta(x)$ denotes the parametrization of $\Gamma_\theta$.
We now introduce the deformation of $P_\epsilon$ on $\Gamma_\theta,\ \theta\in[0,\theta_0)$:
\begin{equation}
\label{eqn:Pepstheta}
    P_{\epsilon,\theta} := P_\theta - i\epsilon (1-\chi(x_\theta)) x_\theta^2,\quad\textrm{with the domain }\widehat{\DD}_\theta := \DD_\theta\cap \widehat{\HH}_\theta.
\end{equation}
It follows from \eqref{P_theta at infinity} that $P_{\epsilon,\theta}$ near infinity is close to the operator 
\begin{equation}
\label{eqn:Hepstheta}
    H_{\epsilon,\theta}:=-e^{-2i\theta}\Delta - i\epsilon e^{2i\theta} x^2,
\end{equation} 
which was considered by Davies \cite{Dav} as an interesting example of a non-normal differential operator. We recall the following basic result:
\begin{lem}
\label{lem:Davies spectrum}
For $\epsilon>0$, $0 \leq \theta < \pi/8$, $H_{\epsilon,\theta}$ is a closed densely defined operator on $L^2(\RR^n)$ equipped with the domain $H^2(\RR^n)\cap \langle x \rangle^{-2}L^2(\RR^n)$. The spectrum is given by
\begin{equation}
\label{eqn:Spectrum Hepstheta}
    \Spec(H_{\epsilon,\theta}) = \{ e^{-i\pi/4} \sqrt{\epsilon}(2|\alpha|+n) : \alpha\in\mathbb{N}_0^n\},\quad |\alpha| := \alpha_1 + \cdots + \alpha_n.
\end{equation}
In addition for any $\delta>0$ we have uniformly in $\epsilon>0$,
\begin{equation}
\label{eqn:Hepstheta resolv bounds}
\begin{gathered}
     (H_{\epsilon,\theta}-z)^{-1}=O_\delta(|z|^{\frac{j-2}{2}}):L^2(\RR^n)\to H^j(\RR^n),\quad j=0,1,2, \\
     \textrm{for }-2\theta+\delta<\arg z<3\pi/2+2\theta-\delta,\quad |z|>\delta.
\end{gathered}
\end{equation}
\end{lem}

\begin{proof}
For every $\epsilon>0$ and $0\leq\theta\leq\theta_0$, $H_{\epsilon,\theta}$ can be viewed as the quantization of the complex-valued quadratic form $q:\RR_x^n\times\RR_\xi^n \to \CC,\quad(x,\xi)\mapsto e^{-2i\theta}\xi^2 - i\epsilon e^{2i\theta}x^2$, which shall be viewed as a closed densely defined operator on $L^2(\RR^n)$ equipped with the domain $\DD(H_{\epsilon,\theta}):=\{u\in L^2(\RR^n):H_{\epsilon,\theta}u\in L^2(\RR^n)\}$. For the analysis of the spectrum for general quadratic operators see Hitrik--Sj\"ostrand--Viola \cite{HSV} and references given there, in particular we obtain \eqref{eqn:Spectrum Hepstheta}. Noticing that the numerical range of $q$ is the sector $\{z\in\CC:3\pi/2 + 2\theta \leq \arg z \leq 2\pi-2\theta\}$, $H_{\epsilon,\theta} - i$ is elliptic with respect to the order function $m=1+x^2+\xi^2$ in the sense that $|q-i|\geq C m$ for some $C=C(\epsilon)>0$. Since $H_{\epsilon,\theta} - i$ is invertible by \eqref{eqn:Spectrum Hepstheta}, we conclude that \[(H_{\epsilon,\theta}-i)^{-1} : L^2(\RR^n)\to m^{-1}(x,D_x) L^2(\RR^n) = H^2(\RR^n)\cap \langle x \rangle^{-2}L^2(\RR^n).\]
Hence $u\in\DD(H_{\epsilon,\theta})\Rightarrow u=(H_{\epsilon,\theta}-i)^{-1} (H_{\epsilon,\theta}u - iu)\in H^2(\RR^n)\cap \langle x \rangle^{-2}L^2(\RR^n) $. Now we rescale $y=\sqrt{\epsilon}x$, then $H_{\epsilon,\theta}$ is unitary equivalent to  $-e^{-2i\theta}\epsilon \Delta_y - i e^{2i\theta}y^2$, that is a semiclassical quadratic operator with $h=\sqrt{\epsilon}$. The bounds \eqref{eqn:Hepstheta resolv bounds} follow from semiclassical ellipticity of $-e^{-2i\theta}\epsilon \Delta_y - i e^{2i\theta}y^2 - z$ for $-2\theta+\delta<\arg z<3\pi/2+2\theta-\delta$, $|z|>\delta$.
\end{proof}

Then we show that $P_{\epsilon,\theta}$ is a Fredholm  operator for $z\notin e^{-i\pi/4}[0,\infty)$.
\begin{lem}
\label{lem:Pepstheta Fredholm}
If $z\in\CC\setminus\{0\}$, $\arg z\neq -\pi/4$, then for each $\epsilon>0$ and $0\leq\theta < \theta_0$, $P_{\epsilon,\theta} - z: \widehat{\DD}_\theta\to \HH_\theta$ is a Fredholm operator of index $0$. In particular the spectrum of $P_{\epsilon,\theta}$ in $\CC\setminus e^{-i\pi/4}[0,\infty)$ is discrete.
\end{lem}

\begin{proof}
We choose $\chi_j\in\CIc(\Gamma_\theta)$, $j=0,1,2,3$, such that $\chi_j=1$ near $\supp\chi_{j-1}$ and that $\chi_0(g_\theta(t)\omega)=1$ for any $t\leq T_0$, thus $1-\chi_j$ are supported in the region where $\Gamma_\theta\owns x_\theta=e^{i\theta}x,\ x\in\RR^n$. Lemma \ref{lem:Davies spectrum} then shows that if $\arg z\neq -\pi/4$, 
\[
    (1-\chi_0)(H_{\epsilon,\theta}-z)^{-1}(1-\chi_1):\HH_\theta \to \widehat{\DD}_\theta.
\]
Now we fix $z\in\CC\setminus\{0\}$ with $\arg z\neq -\pi/4$. Using \eqref{P_theta at infinity} we may assume  that $\supp\chi_0$ is large enough so that 
$\| (P_{\epsilon,\theta}-H_{\epsilon,\theta})(1-\chi_0) (H_{\epsilon,\theta} - z)^{-1} (1-\chi_1)\|_{\HH_\theta\to\HH_\theta} \leq 1/2 $.
Then we choose $z_0=iL$, $L\gg 1$ using \eqref{eqn:Ptheta z0inverse} such that $\| \epsilon(\chi_3 -\chi)x_\theta^2 (P_\theta - z_0)^{-1}\|_{\HH_\theta\to\HH_\theta} \leq 1/2$, thus 
\[
(P_\theta - i\epsilon(\chi_3 -\chi)x_\theta^2 -z_0)^{-1} =  (P_\theta - z_0)^{-1} (I - i\epsilon(\chi_3 -\chi)x_\theta^2 (P_\theta - z_0)^{-1})^{-1}
\] 
exists. We put
\[
    E(z) = \chi_2 (P_\theta -i\epsilon(\chi_3-\chi)x_\theta^2 -z_0)^{-1}\chi_1 + (1-\chi_0)(H_{\epsilon,\theta}-z)^{-1}(1-\chi_1).
\]
Then we get
\[ 
    (P_{\epsilon,\theta} - z) E(z) = I + K(z) + K_1(z), 
\]
where
\[
\begin{gathered}
    K(z) = ((z_0 - z)\chi_2 + [P_\theta,\chi_2]) (P_\theta -i\epsilon(\chi_3-\chi)x_\theta^2 -z_0)^{-1}\chi_1 \\
     + [e^{-2i\theta}\Delta,\chi_0](H_{\epsilon,\theta}-z)^{-1} (1-\chi_1) \\
    K_1(z) = (P_{\epsilon,\theta} - H_{\epsilon,\theta})(1-\chi_0)(H_{\epsilon,\theta} - z)^{-1} (1-\chi_1).
\end{gathered}
\]
Recalling that $\|K_1(z)\|_{\HH_\theta\to\HH_\theta}\leq 1/2$, we obtain that $I+K_1(z)$ is invertible thus
\[
    (P_{\epsilon,\theta} - z)E(z)(I+K_1(z))^{-1} = I + K(z)(I+K_1(z))^{-1}.
\]
Since $(P_\theta - z_0)^{-1}:\HH_\theta\to \DD_\theta$, we conclude that $K(z)$ is compact: $\HH_\theta\to\HH_\theta$. Hence $E(z)(I+K_1(z))$ is an approximate right inverse of $P_{\epsilon,\theta} - z$. 

As an approximate left inverse, we put
\[
    F(z) = \chi_1 (P_\theta -i\epsilon(\chi_3-\chi)x_\theta^2 -z_0)^{-1}\chi_2 + (1-\chi_1)(H_{\epsilon,\theta}-z)^{-1}(1-\chi_0).
\]
Then
\[
    F(z) (P_{\epsilon,\theta} - z) = I + L(z) + L_1(z),
\]
where 
\[
\begin{gathered}
    L(z) = \chi_1 (P_\theta -i\epsilon(\chi_3-\chi)x_\theta^2 -z_0)^{-1} ((z_0 - z)\chi_2 - [P_\theta,\chi_2]) \\
     - (1-\chi_1)(H_{\epsilon,\theta}-z)^{-1}[e^{-2i\theta}\Delta,\chi_0]  \\
    L_1(z) = (1-\chi_1)(H_{\epsilon,\theta} - z)^{-1} (1-\chi_0)(P_{\epsilon,\theta} - H_{\epsilon,\theta}).
\end{gathered}
\]
We may assume again by \eqref{P_theta at infinity} that $\|L_1(z)\|_{\widehat{\DD}_\theta\to\widehat{\DD}_\theta} \leq 1/2$, then 
\[
    (I+L_1(z))^{-1} F(z)(P_{\epsilon,\theta}-z) = I + (I+L_1(z))^{-1} L(z).
\]
Using \eqref{eqn:domain D}, we see that $[e^{-2i\theta}\Delta,\chi_0]$ is compact: $\widehat{\DD}_\theta \to \HH_\theta$, thus $L(z)$ is compact: $\widehat{\DD}_\theta \to \widehat{\DD}_\theta$, $(I+L_1(z))^{-1} F(z)$ is an approximate left inverse.

Since $\|K(z_0)\|_{\HH_\theta\to\HH_\theta}$ and $\|L(z_0)\|_{\widehat{\DD}_\theta \to \widehat{\DD}_\theta}$ can be controlled by the operator norms of $(P-z_0)^{-1}$, $(-\Delta_\theta-z_0)^{-1}$ and $(H_{\epsilon,\theta}-z_0)^{-1}$.
It then follows from \eqref{eqn:Delta_theta resolv bound} and \eqref{eqn:Hepstheta resolv bounds} that $\|K(z_0\|_{\HH_\theta\to\HH_\theta}$, $\|L(z_0)\|_{\widehat{\DD}_\theta \to \widehat{\DD}_\theta} \ll 1$ provided $z_0=iL,\ L\gg 1$, thus $P_{\epsilon,\theta} - iL$ is invertible for $L$ sufficiently large, which implies that $P_{\epsilon,\theta}$ has a discrete spectrum in $\CC\setminus e^{-i\pi/4}[0,\infty)$.
\end{proof}

\begin{lem}
\label{lem:agreement of resolvent}
For each $0\leq \theta < \theta_0$ and $\epsilon>0$, let $\psi\in\CIc(B(0,R_1);[0,1])$ be equal to $1$ near $\overline{B(0,R_0)}$ so that $\psi$ is a function on $\Gamma_\theta$ and defines a multiplication on $\HH_\theta$. Then we have, meromorphically in the region $-\pi/4<\arg z<7\pi/4$,
\begin{equation}
\label{eqn:Pepscutoff resolvent}
    \psi(P_\epsilon - z)^{-1} \psi = \psi (P_{\epsilon,\theta} - z)^{-1}\psi .
\end{equation}    
\end{lem}

\begin{proof}
We modify the proof of \cite[Lemma 2]{Zw-vis}. It is sufficient to show that for $0\leq \theta_1 < \theta_2 < \theta_0$, $|\theta_1 - \theta_2|\ll 1$, 
\begin{equation}
\label{eqn2:Pepscutoff}
    \psi (P_{\epsilon,\theta_1} - z)^{-1}\psi = \psi (P_{\epsilon,\theta_2} - z)^{-1}\psi.
\end{equation}
It is also enough to establish this for $z\in e^{i(-2\theta_1 + \pi/2)}(1,\infty)$ as then the result follows by analytic continuation. For that we show that for $f\in \HH_{R_0}\oplus L^2(B(0,R_1)\setminus B(0,R_0))\subset \HH_{\theta_j}$ there exists $U$ holomorphic in a neighborhood $\Omega_{\theta_1,\theta_2}$ of 
\[ \bigcup_{\theta_1\leq\theta\leq\theta_2} (\Gamma_\theta \setminus B(0,R_0)) \subset \CC^n \]
such that 
\begin{equation}
\label{eqn:holomorphic U}
    U|_{\Gamma_{\theta_j}} (x) = [(P_{\epsilon,\theta_j}-z)^{-1} \psi f](x),\quad\forall\,x\in\Gamma_{\theta_j}\setminus B(0,R_0).
\end{equation}

To show the existence of $U$ such that \eqref{eqn:holomorphic U} holds we apply Lemma \ref{lem:deformation} to a modified family of deformations, which is obtained as follows. Let $\rho\in\CIc((1,6);[0,1])$ be equal to $1$ near $[2,4]$, and put for $T\geq 1$,
\[    
\begin{gathered}
    g_{\theta_1,\theta_2,T}(t) := g_{\theta_1}(t) + \rho(t/T) (g_{\theta_2}(t) - g_{\theta_1}(t)), \\
    \Gamma_{\theta_1,\theta_2,T} := \{ g_{\theta_1,\theta_2,T}(t)\omega : t\in [0,\infty),\ \omega\in\mathbb{S}^{n-1} \}\subset\CC^n . 
\end{gathered}
\]
We can apply Lemma \ref{lem:deformation} to the family of totally real submanifolds interpolating between $\Gamma_{\theta_1}$ and $\Gamma_{\theta_1,\theta_2,T}$, $[0,1]\owns s\mapsto \Gamma_{\theta_1,(1-s)\theta_1 + s\theta_2,T}$. It follows that there exists a holomorphic function $U^T$ defined in a neighborhood of the union of these submanifolds which restricts to $ u_1 := (P_{\epsilon,\theta_1}-z)^{-1} \psi f \in\HH_{\theta_1}$. Varying $T$ we obtain a family of functions agreeing on the intersections of their domains and that gives a holomorphic function $U$ defined in the neighborhood $\Omega_{\theta_1,\theta_2}$.

It remains to show that $U$ restricts to $u_2 \in \HH_{\theta_2}$ (the equation $(P_{\epsilon,\theta_2} - z)u_2 = \psi f$ is automatically satisfied). For $T$ large we put
\[
\begin{gathered}
    \Omega_1(T) = \{ z\in\CC^n : T\leq |z| \leq 6T \}\cap \Gamma_{\theta_1,\theta_2,T} \supset \Gamma_{\theta_1,\theta_2,T}\setminus\Gamma_{\theta_1}, \\
    \Omega_2(T) = \{ z\in\CC^n : T/2 \leq |z| \leq 8T \}\cap \Gamma_{\theta_1,\theta_2,T} ,\quad \Omega_2(T)\setminus\Omega_1(T) \subset e^{i\theta_1}\RR^n,
\end{gathered}
\]
and choose $\chi_T\in\CI(\Omega_2(T);[0,1])$ such that $\chi_T=1$ on $\Omega_1(T)$ with derivative bounds independent of $T$. We recall the following estimate from the proof of \cite[Lemma 3]{Zw-vis}: for $u\in\CI(\Gamma_{\theta_1,\theta_2,T})$, $\tau>1$,
\[
    |\,\langle (-\Delta|_{\Gamma_{\theta_1,\theta_2,T}}-i\epsilon (x|_{\Gamma_{\theta_1,\theta_2,T}})^2 - i e^{-2i\theta_1}\tau )u , u \rangle\,| \geq (\|u\|_{L^2}^2 + \|Du\|_{L^2}^2)/C,
\]
with $C>0$ independent of $\tau,T$, here $\langle\cdot,\cdot\rangle$ is the $L^2$ inner product on $\Gamma_{\theta_1,\theta_2,T}$. Writing 
\[
    P_{\epsilon,\theta_1,\theta_2,T} := P|_{\Gamma_{\theta_1,\theta_2,T}} - i\epsilon (x|_{\Gamma_{\theta_1,\theta_2,T}})^2,
\]
it then follows from \eqref{eqn:Q defn} that
\[ 
    \langle (P_{\epsilon,\theta_1,\theta_2,T} - (-\Delta|_{\Gamma_{\theta_1,\theta_2,T}}-i\epsilon (x|_{\Gamma_{\theta_1,\theta_2,T}})^2 ))u , u \rangle 
    = \int_{\Gamma_{\theta_1,\theta_2,T}} (g^{jk}-\delta^{jk})\partial_k u \partial_j \Bar{u} +c|u|^2.
\]
In view of \eqref{eqn:Qelliptic} and \eqref{Analytic extension}, we obtain that for $T$ sufficiently large,
\[
    |\,\langle (P_{\epsilon,\theta_1,\theta_2,T} - i e^{-2i\theta_1}\tau )\chi_T U , \chi_T U \rangle\,| \geq (\|\chi_T U\|_{L^2}^2 + \|D(\chi_T U)\|_{L^2}^2)/C,
\]
thus $\|\chi_T U\|_{L^2}\leq C \|(P_{\epsilon,\theta_1,\theta_2,T} - i e^{-2i\theta_1}\tau)\chi_T U \|_{L^2}$. We note that 
\[
    (P_{\epsilon,\theta_1,\theta_2,T} - i e^{-2i\theta_1}\tau) U^T = 0 \implies (P_{\epsilon,\theta_1,\theta_2,T} - i e^{-2i\theta_1}\tau)\chi_T U =  [P_{\epsilon,\theta_1,\theta_2,T},\chi_T] U , 
\]
which is supported on $\Omega_2(T)\setminus\Omega_1(T)\subset\Gamma_{\theta_1}$. Hence,
\[
    \|1_{2T\leq |z|\leq 4T}\,u_2\|^2_{L^2(\Gamma_{\theta_2})} \leq  C\|[P_{\epsilon,\theta_1,\theta_2,T},\chi_T] U\|^2_{L^2}\leq C\|1_{T/2 \leq |z|\leq 8T}\,u_1\|^2_{H^1(\Gamma_{\theta_1})}.
\]
We now take $T=2^j$ and sum over $j$, it follows that  $u_2\in\HH_{\theta_2}$.
\end{proof}

\begin{lem}
\label{lem:Pepstheta mult}
For $0\leq \theta < \theta_0$, $\epsilon>0$, the spectrum of $P_{\epsilon,\theta}$ is independent of $\theta$. More precisely, for any $z_0\in\CC\setminus e^{-i\pi/4}[0,\infty)$ we have
\begin{equation}
\label{eqn:Pepstheta mult}
m_{\epsilon,\theta}(z_0) := \rank\oint_{z_0} (P_{\epsilon,\theta} - z)^{-1} dz = \rank \oint_{z_0} \ (P_{\epsilon} - z)^{-1} \, dz,
\end{equation}
where the integral is over a positively oriented circle enclosing $z_0$ and containing no poles other than possibly $z_0$.  
\end{lem}

\begin{proof}
Lemma \ref{lem:Pepstheta Fredholm} shows that 
\begin{equation}
\label{eqn:Proj_epstheta}
    \Pi_{\epsilon,\theta}(z_0) := -\frac{1}{2\pi i}\oint_{z_0} (P_{\epsilon,\theta} - z)^{-1} dz,
\end{equation}
is a finite rank projection which maps $\HH_\theta$ to the generalized eigenspace of $P_{\epsilon,\theta}$ at $z_0$. In view of Lemma \ref{lem:agreement of resolvent}, it suffices to show that for each $0\leq \theta < \theta_0$,
\[
    \rank \Pi_{\epsilon,\theta}(z_0) = \rank \psi \Pi_{\epsilon,\theta}(z_0) \psi.
\]

First we show that $\rank \Pi_{\epsilon,\theta}(z_0) = \rank \Pi_{\epsilon,\theta}(z_0)\psi$, which is equivalent to show that $\rank \psi \Pi_{\epsilon,\theta}(z_0)^* = \rank \Pi_{\epsilon,\theta}(z_0)^*$, since the adjoint of a finite rank operator is of finite rank with the same rank. For that we shall argue by contradiction. Suppose that $\rank \psi \Pi_{\epsilon,\theta}(z_0)^* < \rank \Pi_{\epsilon,\theta}(z_0)^*$, there would exist $0\neq\Tilde{v}\in\Ran \Pi_{\epsilon,\theta}(z_0)^*$ satisfying $\psi \Tilde{v}=0$. Note that $\Pi_{\epsilon,\theta}(z_0)^*$ is also a projection of the form \eqref{eqn:Proj_epstheta} except that  $P_{\epsilon,\theta}^*$ and $\bar{z}_0$ replace $P_{\epsilon,\theta}$ and $z_0$ there, we may assume
\[
    (P_{\epsilon,\theta}^* - \bar{z}_0)^k \Tilde{v} = 0,\quad \Tilde{u}:= (P_{\epsilon,\theta}^* - \bar{z}_0)^{k-1} \Tilde{v}\neq 0,\quad\textrm{for some } k\geq 1.
\]
But that would mean that $\Tilde{u}$ can be identified with an element of $H^2(\Gamma_\theta)$ satisfying
\[
    (Q_{\epsilon,\theta}^* - \bar{z}_0 )\Tilde{u} = 0,\quad \Tilde{u}|_{B(0,R_0)} \equiv 0,\quad Q_{\epsilon,\theta}:=Q_\theta - i\epsilon(1-\chi(x_\theta))x_\theta^2.
\]
Since $Q_{\epsilon,\theta}^*$ is elliptic, unique continuation results
for second order elliptic differential equations -- see H{\"o}rmander \cite[Chapter 17]{hormander} show that $\Tilde{u}\equiv 0$, thus a contradiction.

It remains to show that $\rank \psi\Pi_{\epsilon,\theta}(z_0) \psi= \rank \Pi_{\epsilon,\theta}(z_0)\psi$. Otherwise there would exist solutions $v\in \widehat{\DD}_\theta$ to $(P_{\epsilon,\theta} - z_0)^\ell v =0$, $u:=(P_{\epsilon,\theta} - z_0)^{\ell -1} v \neq 0$ with $\psi v = 0$. It follows that $u$ can be identified with an element of $H^2(\Gamma_\theta)$ satisfying
\[    (Q_{\epsilon,\theta} - z_0)u = 0,\quad u|_{B(0,R_0)} \equiv 0.
\]
Again by the unique continuation results for second order elliptic differential equations, we obtain that $u\equiv 0$, thus a contradiction.
\end{proof}

The next lemma shows that the spectrum of $P_{\epsilon,\theta}$ must stay close to the spectrum of $P_\theta$ when $\epsilon$ is sufficiently small:
\begin{lem}
\label{lem:Pepstheta resolvent norm}
Suppose that $0\leq \theta < \theta_0$ and that $\Omega\Subset \{ z :\,-2\theta < \arg z < 3\pi/2 + 2\theta \}$ is disjoint with $\Spec (P_\theta)$, then there exist $\epsilon_0 = \epsilon_0(\Omega)$ and $C=C(\Omega)$ such that, uniformly in $0<\epsilon<\epsilon_0$, $\Spec(P_{\epsilon,\theta})\cap\Omega=\emptyset$ and
\[    \|(P_{\epsilon,\theta} - z)^{-1}\|_{\HH_\theta\to \DD_\theta} \leq C,\quad z\in\Omega.
\]
\end{lem}

\begin{proof}
We follow closely the proof of \cite[Lemma 5]{Zw-vis} except that $P_\theta$ replaces $-\Delta_\theta$ there. Let $\chi_j\in\CIc([0,\infty);[0,1])$ be equal to 1 on $[0,R_0]$ and satisfy $\chi_j=1$ near $\supp\chi_{j-1}$, $j=1,2$. Parametrizing  $\Gamma_\theta$ by $f_\theta:[0,\infty)\times\mathbb{S}^{n-1}\owns (t,\omega)\mapsto g_\theta(t)\omega\in\Gamma_\theta$, we define functions $\chi_j^h\in\CIc(\Gamma_\theta)$ by
\[
    \chi_j^h(g_\theta(t)\omega) := \chi_j(th),\quad 0<h\leq 1.
\]
For $z\in\Omega$ we put
\[
    E_{\epsilon,\theta}^h(z) := \chi_2^h (P_\theta - z)^{-1} \chi_1^h + (1-\chi_0^h)(H_{\epsilon,\theta} -z)^{-1}(1-\chi_1^h),
\]
so that $(P_{\epsilon,\theta}-z)E_{\epsilon,\theta}^h(z) = I + K_{\epsilon,\theta}^h(z)$, where
\begin{align*}
    K_{\epsilon,\theta}^h(z) := &- i\epsilon (1-\chi)x_\theta^2\chi_2^h (P_\theta - z)^{-1} \chi_1^h + [P_\theta,\chi_2^h] (P_\theta - z)^{-1} \chi_1^h \\ 
    & + (P_{\epsilon,\theta}-H_{\epsilon,\theta})(1-\chi_0^h)(H_{\epsilon,\theta} -z)^{-1}(1-\chi_1^h) \\
    & -[P_\theta,\chi_0^h](1-\chi_0^h)(H_{\epsilon,\theta} -z)^{-1}(1-\chi_1^h).
\end{align*}
Using \eqref{P_theta at infinity} and \eqref{eqn:Hepstheta resolv bounds} we see that for $h$ small enough,
\[ \|(P_{\epsilon,\theta}-H_{\epsilon,\theta})(1-\chi_0^h)(H_{\epsilon,\theta} -z)^{-1}(1-\chi_1^h)\|_{L^2(\Gamma_\theta)\to L^2(\Gamma_\theta)} < 1/4. \]
Since $[Q_\theta,\chi_j^h]=O(h):H^1(\Gamma_\theta)\to L^2(\Gamma_\theta)$ and $x_\theta^2 \chi_2^h=O(h^{-2}):L^2(\Gamma_\theta)\to L^2(\Gamma_\theta)$, we can first choose $h$ sufficiently small then there exists $\epsilon_0=\epsilon_0 (h,\Omega)$ such that for all $\epsilon<\epsilon_0(h,\Omega)$ and $z\in\Omega$,  $\|K_{\epsilon,\theta}^h(z)\|_{\HH_\theta\to\HH_\theta}<1/2$, thus $I+K_{\epsilon,\theta}^h(z)$ has a uniformly bounded inverse and $(P_{\epsilon,\theta}-z)^{-1} = E_{\epsilon,\theta}^h (z) (I+K_{\epsilon,\theta}^h (z))^{-1}$ exists. It follows from \eqref{eqn:Hepstheta resolv bounds} that there exists $C=C(\Omega)$ independent of $\epsilon$ such that for $z\in\Omega$, $\|E_{\epsilon,\theta}^h (z)\|_{\HH_\theta\to\DD_\theta} \leq C$, which completes the proof.
\end{proof}

\section{The obstacle problem and the Dirichlet-to-Neumann operator}
\label{section:N operator}
In the black box case we cannot use the strategy of \cite{Zw-vis} which covers the case $P=-\Delta+V$, $V\in L_{\comp}^\infty$. Instead we introduce an artificial obstacle to separate the abstract black box from the differential operator outside. By an {\em obstacle} we mean an open set $\OO$ with smooth boundary as in \S\ref{reference operator}. Suppose that $\OO$ contains $\overline{B(0,R_0)}$ and that $\chi$ in \eqref{eqn:Peps} be equal to 1 near $\overline{\OO}$. Let $\nu(x)$ be the Euclidean normal vector of $\partial\OO$ at $x$ pointing into $\OO$, we put
\begin{equation}
\label{eqn:nu_g}
    \nu_g(x) := (g^{jk}(x))_{n\times n} \cdot \nu(x),\quad x\in\partial\OO.
\end{equation}

First we introduce the interior Dirichlet-to-Neumann operator of $P$:
\begin{equation}
\label{eqn:N in}
    \NN_P^\In (z) \varphi := \frac{\partial u}{\partial\nu_g},\quad\textrm{where $u$ solves the problem }
    \begin{gathered}
        (P-z)u = 0\textrm{ in } \OO \\
        u = \varphi\textrm{ on }\partial\OO
    \end{gathered}.
\end{equation}
$\NN_P^\In(z)$ is well-defined once we establish the existence and uniqueness of the solution $u$ to the boundary-value problem in \eqref{eqn:N in}. This can be done if $z$ is not an eigenvalue of the operator $P^\OO$ introduced in \S \ref{reference operator}. Indeed, we set $E^{\textrm{in}}:H^{3/2}(\partial\OO)\to H^2(\OO)$ as a linear bounded extension operator such that $E^{\textrm{in}}\varphi |_{\partial\OO}=\varphi$ and $\supp E^{\textrm{in}}\varphi\subset\overline{\OO}\setminus B(0,R_0)$ for any $\varphi$. Then for $z\notin\Spec(P^\OO)$, $u= E^{\textrm{in}}\varphi - (P^\sharp - z)^{-1}(Q-z)E^{\textrm{in}}\varphi$ is the unique solution to \eqref{eqn:N in}, we obtain that
\begin{equation}
\label{eqn:N in with extension}
    \NN_P^{\textrm{in}}(z) \varphi = \partial_{\nu_g} (E^{\textrm{in}}\varphi - (P^\sharp - z)^{-1}(Q-z)E^{\textrm{in}}\varphi),
\end{equation}
Hence $z\mapsto \NN_P^\In (z):H^{3/2}(\partial\OO) \to H^{1/2}(\partial\OO)$ is a meromorphic family of operators on $\CC$ with poles contained in $\Spec(P^\OO)$.

Similarly, we can define the exterior Dirichlet-to-Neumann operator of $P_{\epsilon,\theta}$ for every $0\leq \theta < \theta_0$ and $\epsilon\geq 0$:
\begin{equation}
\label{eqn:Nepstheta out}
    \NN_{\epsilon,\theta}^{\textrm{out}}(z) \varphi := \frac{\partial u}{\partial\nu_g},\quad\textrm{where $u$ solves the problem }
    \begin{gathered}
        (Q_{\epsilon,\theta}-z)u = 0\textrm{ in } \Gamma_\theta\setminus\OO \\
        u = \varphi\textrm{ on }\partial\OO
    \end{gathered}.
\end{equation}
To show the well-definedness of $\NN_{\epsilon,\theta}^\out(z)$, we introduce the restriction of $Q_{\epsilon,\theta}$ to $\Gamma_\theta\setminus\OO$ with Dirichlet boundary condition as follows:
\begin{equation}
\label{eqn:Qepstheta Obstacle}
\begin{gathered}
    Q_\theta^\OO: H^2(\Gamma_\theta\setminus\OO)\cap H_0^1(\Gamma_\theta\setminus\OO) \to L^2(\Gamma_\theta\setminus\OO),\quad Q_\theta^\OO u := Q_\theta u, \\
    Q_{\epsilon,\theta}^\OO:=Q_\theta^\OO - i\epsilon(1-\chi)x_\theta^2\quad\textrm{with domain }\DD(Q_\theta^\OO)\cap |x_\theta|^{-2}L^2(\Gamma_\theta\setminus\OO).
\end{gathered}
\end{equation}
Since $Q_\theta^\OO$ and $Q_{\epsilon,\theta}^\OO$ can also be viewed as black box perturbations of $-\Delta_\theta$ and $H_{\epsilon,\theta}$ respectively, we conclude from Lemma \ref{lem:Ptheta Fredholm} and Lemma \ref{lem:Pepstheta Fredholm} that
$Q_{\epsilon,\theta}^\OO - z,\ \epsilon\geq 0$ is a Fredholm operator of index $0$ for $-2\theta<\arg z<3\pi/2 +2\theta$. We claim that $\NN_{\epsilon,\theta}^\out(z)$ is well defined if $z\notin\Spec(Q_{\epsilon,\theta}^\OO)$. For that let $E^\out: H^{3/2}(\partial\OO)\to H^2(\Gamma_\theta\setminus \OO)$ be a linear bounded extension operator with $E^\out \varphi |_{\partial\OO} = \varphi$ and $\supp E^\out \varphi\Subset \Gamma_\theta\setminus\OO$, then 
\begin{equation}
\label{eqn:N out with extension}
    \NN_{\epsilon,\theta}^\out (z) \varphi = \partial_{\nu_g} (E^\out \varphi - (Q_{\epsilon,\theta}^\OO - z)^{-1}(Q_{\epsilon,\theta}-z)E^\out \varphi).    
\end{equation}
It follows that $z\mapsto \NN_{\epsilon,\theta}^\out (z):H^{3/2}(\partial\OO) \to H^{1/2}(\partial\OO)$ is a meromorphic family of operators in the region $-2\theta<\arg z<3\pi/2 +2\theta$, with poles contained in $\Spec(Q_{\epsilon,\theta}^\OO)$.

Now we put
\begin{equation}
\label{eqn:Nepstheta}    
    \NN_{\epsilon,\theta}(z) := \NN_{\epsilon,\theta}^{\textrm{out}}(z) - \NN_P^{\textrm{in}}(z).
\end{equation}

\begin{lem}
\label{lem:Nepstheta Fredholm}
Suppose that $0\leq \theta < \theta_0$, $\epsilon\geq 0$ and that $-2\theta<\arg z<3\pi/2 +2\theta$ with $z\notin \Spec(P^\OO)\cup\Spec(Q_{\epsilon,\theta}^\OO)$, then $\NN_{\epsilon,\theta}(z):H^{3/2}(\partial\OO) \to H^{1/2}(\partial\OO)$ is a Fredholm operator of index $0$.
\end{lem}

\begin{proof}
Let $Q_\In^\OO$ and $\NN_Q^\In (z)$ be the the reference operator and the interior Dirichlet-to-Neumann operator associated with $Q$,  defined as in \eqref{eqn:P sharp} and \eqref{eqn:N in} respectively except that $Q$ replaces $P$ there. Choosing $z_0\notin\Spec(Q_\In^\OO)\cup\Spec(Q_{\epsilon,\theta}^\OO)\cup\Spec(Q_{\epsilon,\theta})$, we claim that 
\begin{equation}
\label{eqn:N Qepstheta invertible}
    \NN_{\epsilon,\theta}^{\textrm{out}}(z_0) - \NN_Q^{\textrm{in}}(z_0) : H^{3/2}(\partial\OO) \to H^{1/2}(\partial\OO)\quad\textrm{is invertible.}
\end{equation}
   
To show injectivity, we argue by contradiction. Suppose that $0\neq\varphi\in H^{3/2}(\partial\OO)$ satisfies $\NN_{\epsilon,\theta}^{\textrm{out}}(z_0)\varphi = \NN_Q^{\textrm{in}}(z_0)\varphi$, it follows from \eqref{eqn:N in} and \eqref{eqn:Nepstheta out} that there exist $u_1\in H^2(\OO)$, $u_2\in H^2(\Gamma_\theta\setminus\OO)$ ($|x_\theta|^2 u_2\in L^2(\Gamma_\theta\setminus\OO)$ when $\epsilon>0$) such that
\begin{equation}
\label{eqn:u1u2 for Qepstheta}
\textrm{$u_1$ solves }
        \begin{gathered}
        (Q-z_0)u_1 = 0\textrm{ in } \OO \\
        u_1 = \varphi\textrm{ on }\partial\OO
        \end{gathered},\textrm{ and $u_2$ solves }
        \begin{gathered}
        (Q_{\epsilon,\theta}-z_0)u_2 = 0\textrm{ in } \Gamma_\theta\setminus\OO \\
        u_2 = \varphi\textrm{ on }\partial\OO
        \end{gathered},
\end{equation}
and that $\partial_{\nu_g}u_1 = \partial_{\nu_g}u_2$.
Let $u=1_{\OO}\, u_1 + 1_{\Gamma_\theta \setminus\OO}\,u_2$, we aim to show that $u\in H^2(\Gamma_\theta)$. For that it suffices to show the regularity of $u$ near $\partial\OO$. For any $x_0\in\partial\OO$, we choose $B_{x_0} := B(x_0,r)\subset B(0,R_1)$ such that $Q_{\epsilon,\theta}=Q$ in $B_{x_0}$ and put $v\in\CIc(B_{x_0})$. Then we integrate by parts to obtain:
\[
\begin{split}
    &{\qquad}\int_{B_{x_0}} \left(\sum_{j,k=1}^n g^{jk} \partial_{x_k}u\, \partial_{x_j} v + cuv\right)dx \\
    &= \int_{B_{x_0}\cap\OO}\left(\sum_{j,k=1}^n g^{jk} \partial_{x_k}u_1 \partial_{x_j} v + cu_1 v\right)dx + \int_{B_{x_0}\setminus\OO}\left(\sum_{j,k=1}^n g^{jk} \partial_{x_k}u_2 \partial_{x_j} v + cu_2 v\right)dx \\
    & = \int_{B_{x_0}\cap\OO} v\,Q u_1 dx - \int_{\partial\OO\cap B_{x_0}}v\,\partial_{\nu_g}u_1 dS(x) +\int_{B_{x_0}\setminus\OO} v\,Q u_2 dx + \int_{\partial\OO\cap B_{x_0}} v\,\partial_{\nu_g}u_1 dS(x) \\
    & = \int_{B_{x_0}\cap\OO} z_0 u_1 v \,dx + \int_{B_{x_0}\setminus\OO} z_0 u_2 v \,dx = \int_{B_{x_0}} z_0 uv \,dx.
\end{split}
\]
Hence $u$ is a weak solution of $(Q-z_0)u=0$ in $B_{x_0}$, the regularity results
for second order elliptic differential equations show that $u$ is $H^2$ near $x_0$, thus $u\in H^2(\Gamma_\theta)$. It then follows from \eqref{eqn:u1u2 for Qepstheta} that $u$ solves the equation $(Q_{\epsilon,\theta}-z_0)u=0$, thus $z_0\in\Spec(Q_{\epsilon,\theta})$, which gives a contradiction.

To show surjectivity, we first choose a linear bounded operator $L_g:H^{1/2}(\partial\OO)\to H^2(\OO)$ satisfying the following:
\begin{equation}
\label{eqn:operator Lg}
\begin{gathered}
    L_g \Tilde{\varphi} := v,\quad\textrm{where }v\in H^2(\OO)\cap H_0^1(\OO)\  \textrm{satisfies} \\
    \supp v\subset \overline{\OO}\setminus B(0,R_0)\ \textrm{and }\partial_{\nu_g}v = \Tilde{\varphi},\quad\Tilde{\varphi}\in H^{1/2}(\partial\OO).
\end{gathered}
\end{equation}
For any $\Tilde{\varphi}\in H^{1/2}(\partial\OO)$, let $v:=L_g \Tilde{\varphi}$, $f:=(Q_\In^\OO -z_0)v\in L^2(\OO)$ and we put 
\[
u:= (Q_{\epsilon,\theta} - z_0)^{-1} \imath f\quad\textrm{and}\quad  \varphi:=u|_{\partial\OO}\in H^{3/2}(\OO),
\]
where $\imath:L^2(\OO)\hookrightarrow L^2(\Gamma_\theta)$ denotes the extension by zero. Then $u_1:=1_\OO u - v$ solves the boundary value problem $(Q-z_0)u_1 = 0$ in $\OO$, $u_1=\varphi$ on $\partial\OO$; $u_2:= 1_{\Gamma_\theta\setminus\OO}\,u$ solves $(Q_{\epsilon,\theta}-z_0)u_2 = 0$ in $\Gamma_\theta\setminus\OO$, $u_2=\varphi$ on $\partial\OO$. Hence we have
\[
\NN_{\epsilon,\theta}^{\textrm{out}}(z_0)\varphi - \NN_Q^{\textrm{in}}(z_0)\varphi = \partial_{\nu_g}1_{\Gamma_\theta\setminus\OO}\,u - \partial_{\nu_g}(1_\OO u - v) = \partial_{\nu_g}v = \Tilde{\varphi}.
\]

In view of \eqref{eqn:N Qepstheta invertible}, it now suffices to show that
$\NN_{\epsilon,\theta}^{\textrm{out}}(z)-\NN_{\epsilon,\theta}^{\textrm{out}}(z_0)$ and $\NN_P^{\textrm{in}}(z) - \NN_Q^{\textrm{in}}(z_0)$ are compact: $H^{3/2}(\partial\OO)\to H^{1/2}(\partial\OO)$. Using \eqref{eqn:N out with extension} we have for any $\varphi\in H^{3/2}(\OO)$,
\[
\begin{split}
    &{\ \quad} \NN_{\epsilon,\theta}^\out (z)\varphi - \NN_{\epsilon,\theta}^\out (z_0)\varphi \\
    & = \partial_{\nu_g}((Q_{\epsilon,\theta}^\OO - z_0)^{-1} (Q_{\epsilon,\theta}-z_0) - (Q_{\epsilon,\theta}^\OO - z)^{-1} (Q_{\epsilon,\theta}-z)) E^\out \varphi \\
    & = (z-z_0) \partial_{\nu_g} (Q_{\epsilon,\theta}^\OO - z_0)^{-1} (I - (Q_{\epsilon,\theta}^\OO - z)^{-1}(Q_{\epsilon,\theta}-z)) E^\out \varphi \in H^{5/2}(\partial\OO),
\end{split}
\]
thus $\NN_{\epsilon,\theta}^{\textrm{out}}(z)-\NN_{\epsilon,\theta}^{\textrm{out}}(z_0): H^{3/2}(\partial\OO)\to H^{5/2}(\partial\OO) \subset H^{1/2}(\partial\OO)$ is compact since the last inclusion map is compact. It remains to show that $\NN_P^{\textrm{in}}(z) - \NN_Q^{\textrm{in}}(z_0)$ is compact: $H^{3/2}(\partial\OO)\to H^{1/2}(\partial\OO)$. Let $\psi\in\CIc(\OO)$ be equal to $1$ near $\overline{B(0,R_0)}$, $\varphi\in H^{1/2}(\OO)$, there exist $u$ and $v$ satisfying:
\[
    \begin{gathered}
        (P-z)u = 0\textrm{ in } \OO \\
        u = \varphi\textrm{ on }\partial\OO
    \end{gathered}\quad\textrm{and}\quad
    \begin{gathered}
        (Q-z_0)v = 0\textrm{ in } \OO \\
        v = \varphi\textrm{ on }\partial\OO
    \end{gathered},
\]
recalling \eqref{eqn:D sharp} that $(1-\psi)u\in H^2(\OO)$, thus we have 
\[
    (\NN_P^{\textrm{in}}(z) - \NN_Q^{\textrm{in}}(z_0))\varphi = \partial_{\nu_g}((1-\psi)u - (1-\psi)v).
\]
Using \eqref{eqn:Q defn} we can show that $(1-\psi)u-(1-\psi)v\in H^2(\OO)\cap H_0^1(\OO)$ satisfies:
\[
\begin{split}
    Q((1-\psi)u - (1-\psi)v) & = (1-\psi)Pu - [P,\psi]u - (1-\psi)Qv + [Q,\psi]v \\
    & = z(1-\psi)u - z_0(1-\psi)v - [P,\psi]u + [Q,\psi]v \in H^1(\OO),
\end{split}
\]
then we conclude from the regularity results for second order elliptic differential equations that $(1-\psi)u-(1-\psi)v\in H^3(\OO)$, thus $(\NN_P^{\textrm{in}}(z) - \NN_Q^{\textrm{in}}(z_0))\varphi\in H^{3/2}(\partial\OO)$. Therefore, $\NN_P^{\textrm{in}}(z) - \NN_Q^{\textrm{in}}(z_0): H^{3/2}(\partial\OO)\to H^{3/2}(\partial\OO)\subset H^{1/2}(\partial\OO)$ is compact, which completes the proof.
\end{proof}

\noindent
{\bf Remark:} The compactness of $\NN_{\epsilon,\theta}^{\textrm{out}}(z)-\NN_{\epsilon,\theta}^{\textrm{out}}(z_0)$ and $\NN_P^{\textrm{in}}(z) - \NN_Q^{\textrm{in}}(z_0)$ can also be proved using the facts that the principal symbols of
$\NN_{\epsilon,\theta}^{\textrm{out}}(z)$ and $\NN_{\epsilon,\theta}^{\textrm{out}}(z_0)$ are identical, same for $\NN_P^{\textrm{in}}(z)$ and $\NN_Q^{\textrm{in}}(z_0)$ -- see for instance Lee--Uhlmann \cite{LeUh1989} for a detailed account.

In order to work on a single Hilbert space, we introduce 
\begin{equation}
\label{eqn:hat NNepstheta}
    \widehat{\NN}_{\epsilon,\theta}(z) := \langle D_{\partial\OO} \rangle^{-1} \NN_{\epsilon,\theta}(z) : H^{3/2}(\partial\OO)\to H^{3/2}(\partial\OO),
\end{equation}
where $\langle D_{\partial\OO} \rangle = (1-\Delta_{\partial\OO})^{1/2}$ is the standard isomorphism between Sobolev spaces $H^s(\partial\OO)$ and $H^{s-1}(\partial\OO)$. Now we are ready to state the main results of this section:
\begin{lem}
\label{lem:DtoN}
Suppose that $0\leq \theta < \theta_0$, $\epsilon\geq 0$ and that $\Omega\Subset \{ z : -2\theta < \arg z < 3\pi/2 + 2\theta \}$ is  disjoint from $\Spec(P^\sharp)\cup\Spec(Q_{\epsilon,\theta}^\OO)$, 
\[
    z \mapsto \widehat{\NN}_{\epsilon,\theta}(z)^{-1},\quad z\in\Omega,
\]
is a meromorphic family of operators on $H^{3/2}(\partial\OO)$ with poles of finite rank. Moreover,
\begin{equation}
\label{eqn:charaterize eigenvalues}
    n_{\epsilon,\theta}(z) := \frac{1}{2\pi i} \tr\oint_z \widehat{\NN}_{\epsilon,\theta}(w)^{-1}\partial_w \widehat{\NN}_{\epsilon,\theta}(w)\,dw = m_{\epsilon,\theta}(z),
\end{equation}    
where the integral is over a positively oriented circle enclosing $z$ and containing no
poles other than possibly $z$ and $m_{\epsilon,\theta}(z)$ is given by \eqref{eqn:Pepstheta mult} (and by \eqref{eqn:multiplicity z0} when $\epsilon=0$). 
\end{lem}

\begin{proof}
1. Suppose that $z\in\Omega$ is an eigenvalue of $P_{\epsilon,\theta}$, we choose $u\in\ker(P_{\epsilon,\theta} -z)$ and let $\varphi=u|_{\partial\OO}$, then $\NN_{\epsilon,\theta}^{\textrm{out}}(z)\varphi - \NN_P^{\textrm{in}}(z)\varphi = \partial_{\nu_g} u - \partial_{\nu_g}u = 0$. Note that $\varphi\neq 0$ since $z\notin\Spec(P^\sharp)$, thus $\ker \widehat{\NN}_{\epsilon,\theta}(z)\neq\{0\}$. On the other hand, suppose that $0\neq \varphi\in\ker\widehat{\NN}_{\epsilon,\theta}(z)$, the same arguments as in the proof of Lemma \ref{lem:Nepstheta Fredholm} show that $z\in\Spec(P_{\epsilon,\theta})$. Hence
\begin{equation}
\label{eqn:eigenvalues as zeros}
    z\in\Spec(P_{\epsilon,\theta}) \Longleftrightarrow \ker \widehat{\NN}_{\epsilon,\theta}(z)\neq \{0\},
\end{equation}
and we conclude from Lemma \ref{lem:Nepstheta Fredholm} that $\widehat{\NN}_{\epsilon,\theta}(z)$ is invertible for $z\in\Omega\setminus\Spec(P_{\epsilon,\theta})$. Analytic Fredholm theory then shows that $\Omega\owns z\mapsto \widehat{\NN}_{\epsilon,\theta}(z)^{-1}$ is a meromorphic family of operators on $H^{3/2}(\partial\OO)$ with poles of finite rank.

2. Since \eqref{eqn:eigenvalues as zeros} proves \eqref{eqn:charaterize eigenvalues} in the case $m_{\epsilon,\theta}(z)=0$, we now assume that $m_{\epsilon,\theta}(z)=M\geq 1$, and that $P_{\epsilon,\theta}$ has exactly one eigenvalue $z$ in $D(z,2r):=\{\zeta\in\CC,|\zeta-z|<2r\}$. Since $\Omega\cap\Spec(P^\sharp)=\emptyset$, $z$ is not a compactly supported embedded eigenvalue of $P$, that is, there does not exist $0\neq u\in \DD$ with $\supp u\subset B(0,R_0)$ such that $(P-z)u=0$. We claim that for any $\delta>0$ there exists $V\in \CI(\OO\setminus B(0,R_0);\RR)$ with $\|V\|_\infty < \delta$ such that 
\[ 
    \rank\int_{\partial D(z,r)}(P_{\epsilon,\theta}+V-w)^{-1} dw = M,  
\]
and that the eigenvalues of $P_{\epsilon,\theta}+V$ in $D(z,r)$ are all simple. This follows from the results of Klopp--Zworski \cite{klopp} (see also \cite[Theorem 4.39]{res}) and we omit the proof here.
Replacing $P$ by $P+V$ in \eqref{eqn:N in}, we can define $\widehat{\NN}_{\epsilon,\theta}^V$ for $P_{\epsilon,\theta}+V$ as in \eqref{eqn:Nepstheta} and \eqref{eqn:hat NNepstheta}.
Note that $\widehat{\NN}_{\epsilon,\theta}$ has no kernel except at $z$ in $D(z,2r)$ by \eqref{eqn:eigenvalues as zeros}, using \eqref{eqn:N in with extension} we can choose $\delta$ small enough such that for $\|V\|_\infty < \delta$,
\[
    \| \widehat{\NN}_{\epsilon,\theta}(w)^{-1} (\widehat{\NN}_{\epsilon,\theta}(w) - \widehat{\NN}_{\epsilon,\theta}^V(w)) \|_{H^{3/2}(\OO)\to H^{3/2}(\OO)} < 1,\quad \forall w\in\partial D(z,r).
\]
It then follows from the Gohberg--Sigal--Rouch\'e theorem (see Gohberg--Sigal \cite{gohberg1971operator} and \cite[Appendix C]{res}) that
\[
    \frac{1}{2\pi i} \tr\int_{\partial D(z,r)} \NN_{\epsilon,\theta}^V (w)^{-1}\partial_w \NN_{\epsilon,\theta}^V (w)\,dw = n_{\epsilon,\theta}(z).
\]
Hence it is enough to prove \eqref{eqn:charaterize eigenvalues} in the case $m_{\epsilon,\theta}(z)=1$ with $P_{\epsilon,\theta}$ replaced by $P_{\epsilon,\theta}+V$.

3. Now we assume that $m_{\epsilon,\theta}(z)=1$. In view of \eqref{eqn:eigenvalues as zeros}, $\widehat{\NN}_{\epsilon,\theta}(w)^{-1}$ has a pole at $z$, it remains to show that $z$ is a simple pole. For any $w$ near $z$ and $\Tilde{\varphi}\in H^{1/2}(\partial\OO)$, we recall \eqref{eqn:operator Lg} that $L_g \Tilde{\varphi}\in \DD^\sharp$, then $(P^\sharp - w)L_g \Tilde{\varphi}\in\HH^\sharp$. Now we put
\[
    u:=(P_{\epsilon,\theta}-w)^{-1} \imath (P^\sharp - w)L_g \Tilde{\varphi},\quad \varphi:=u|_{\partial\OO},
\]
where $\imath:\HH^\sharp\hookrightarrow \HH_\theta$ is the extension by zero. Following the arguments in the proof of Lemma \ref{lem:Nepstheta Fredholm} while $P$ replacing $Q$ there, we can show that $\NN_{\epsilon,\theta} (w)\varphi = \Tilde{\varphi}$, thus
\[
    \widehat{\NN}_{\epsilon,\theta} (w)^{-1}\Tilde{\varphi} = ((P_{\epsilon,\theta}-w)^{-1} \imath (P^\sharp - w)L_g (\langle D_{\partial\OO}\rangle \Tilde{\varphi}))|_{\partial\OO},\quad \forall\Tilde{\varphi}\in H^{3/2}(\partial\OO).
\]
Since $z$ is a simple pole of $w\mapsto (P_{\epsilon,\theta}-w)^{-1}$ by our assumptions, it follows from the expression above that $z$ must be a simple pole of $w\mapsto\widehat{\NN}_{\epsilon,\theta} (w)^{-1}$.
\end{proof}

\section{Deformation of obstacles}
\label{section:deform obstacle}
We have shown that the eigenvalues of $P_{\epsilon,\theta},\ \epsilon\geq 0$, can be identified with the poles of $z\mapsto \NN_{\epsilon,\theta}(z)^{-1}$. One problem of this characterization is that $\NN_{\epsilon,\theta}(z)$ can only be defined away from $\Spec(P^\sharp)$ and $\Spec(Q_{\epsilon,\theta}^\OO)$. In this section we will show that the spectrum of $P^\sharp$ and $Q_{\theta}^\OO$ can be moved by deforming the obstacle $\OO$. Hence for any resonance $z_0$ of $P$, we can always assume that $\NN_\theta(z)$ is well-defined in some neighborhood of $z_0$ by selecting a proper obstacle.  

To describe the deformations of obstacles, we follow Pereira \cite{pereira} and introduce a set of smooth mappings which deforms the obstacle $\OO$:
\begin{equation}
\label{def:Diffeomorphism}
    \Diff(\OO) := \left\{ 
        \begin{gathered}
        \Phi\in \CI(\RR^n;\RR^n) \textrm{ is a diffeomorphism}:\ \Phi(\partial\OO) = \partial\,\Phi(\OO), \\ 
        \Phi(x)=x,\quad\textrm{for all } |x|\leq R_0\ \textrm{or } |x|\geq R_1.   
        \end{gathered}
    \right\}
\end{equation}
We note that $\Phi\in\Diff(\OO)$ only deforms the region $\{x\in\RR^n:R_0 < |x| < R_1\}$, then it also defines a diffeomorphism of $\Gamma_\theta$, $0\leq \theta < \theta_0$. The pullback $\Phi^*$ gives an isomorphism between $L^2(\Gamma_\theta\setminus\Phi(\OO))$ and $L^2(\Gamma_\theta\setminus\OO)$, which also restricts to an isomorphism between $\DD(Q_\theta^{\Phi(\OO)})$ and $\DD(Q_\theta^\OO)$ given in \eqref{eqn:Qepstheta Obstacle} since it preserves the Dirichlet boundary condition. Hence we can define the deformed operator of $Q_\theta^\OO$ associated with the deformation $\Phi$ as follows:
\begin{equation}
\label{eqn:QthetaPhiOO}
    Q_{\theta,\Phi}^\OO := \Phi^* Q_\theta^{\Phi(\OO)} (\Phi^*)^{-1},\quad\textrm{with }\DD(Q_{\theta,\Phi}^\OO)=\DD(Q_\theta^\OO).
\end{equation}
The Fredholm properties of $Q_\theta^{\Phi(\OO)}-z$ immediately show that $Q_{\theta,\Phi}^\OO - z $ is a Fredholem operator of index $0$ for $-2\theta < \arg z < 3\pi/2 + 2\theta$, and \eqref{eqn:QthetaPhiOO} implies that the spectrum of $Q_{\theta,\Phi}^\OO$ in this region is identical to the spectrum of $Q_\theta^{\Phi(\OO)}$. Moreover, $Q_{\theta,\Phi}^\OO$ can be viewed as a restriction of $Q_{\theta,\Phi}:=\Phi^* Q_\theta (\Phi^*)^{-1}$ to $\Gamma_\theta\setminus\OO$ with Dirichlet boundary condition. A direct calculation shows that
\begin{equation}
\label{eqn:A Phi}
    A_\Phi := \Phi^* Q_\theta (\Phi^*)^{-1} - Q_\theta = \Phi^* Q (\Phi^*)^{-1} - Q = \sum_{|\alpha|\leq 2} a_\alpha(x) \partial_x^\alpha,
\end{equation}
where the coefficients $a_\alpha$ are supported in $B(0,R_1)\setminus\overline{B(0,R_0)} \subset \Gamma_\theta$. We note that $\|a_\alpha\|_\infty \leq C \|\Phi-\id\|_{C^2}$, thus $A_\Phi = \OO(\|\Phi-\id\|_{C^2}):H^2(\Gamma_\theta) \to L^2(\Gamma_\theta)$. 

Now we show that $\Spec(Q_\theta^\OO)$ can be moved by deforming the obstacle:
\begin{lem} 
\label{lem:deform obstacle}
Suppose that the obstacle $\OO\subset B(0,R_1)$ contains $\overline{B(0,R_0)}$ and that $-2\theta < \arg z_0 < 3\pi/2 + 2\theta$, then for any $\delta>0$ there exists $\Phi\in\Diff(\OO)$ with $\|\Phi-\id\|_{C^2} < \delta$ such that $z_0 \notin \Spec(Q_{\theta}^{\Phi(\OO)})$.
\end{lem}

\begin{proof}
We may assume that $z_0\in\Spec(Q_\theta^\OO)$, otherwise we can take $\Phi=\id$. Suppose that $Q_\theta^\OO$ has exactly one eigenvalue in $D(z_0,2r)$. For $D:=D(z_0,r)$ we define
\begin{equation}
\label{eqn:Proj QthetaOO}
    \Pi_\OO(D) := -\frac{1}{2\pi i} \int_{\partial D} (Q_\theta^\OO - \zeta)^{-1} d\zeta,\quad m_\OO(D):=\rank\Pi_\OO(D),
\end{equation}
then $m_\OO(D)=m_\OO(z_0)$, where $m_\OO(z_0)$ denotes the multiplicity of $z_0\in\Spec(Q_\theta^\OO)$.

For $\delta>0$ small, we put
\[    
\UU_\delta(\OO) := \{\Phi\in\Diff(\OO) : \|\Phi-\id\|_{C^2(\RR^n\setminus\OO)} < \delta\}.
\]
It follows from \eqref{eqn:A Phi} that $Q_{\theta,\Phi}^\OO - Q_\theta^\OO = O(\|\Phi-\id\|_{C^2}): H^2(\Gamma_\theta\setminus\OO)\to L^2(\Gamma_\theta\setminus\OO)$, thus for $\Phi\in\UU_\delta(\OO)$ with $\delta$ sufficiently small, 
\[
    (Q_{\theta,\Phi}^\OO - \zeta)^{-1} = (Q_{\theta}^\OO - \zeta)^{-1} (I + (Q_{\theta,\Phi}^\OO - Q_\theta^\OO)(Q_{\theta}^\OO - \zeta)^{-1})^{-1},\quad\zeta\in\partial D,
\]
exists and $\sup_{\zeta\in\partial D}\|(Q_{\theta,\Phi}^\OO - \zeta)^{-1} - (Q_\theta^\OO -\zeta)^{-1}\|_{L^2(\Gamma_\theta\setminus\OO)\to L^2(\Gamma_\theta\setminus\OO)} < C(\Omega)\delta $. We define
\begin{equation}
\label{eqn:Proj QthetaPhiOO}
    \Pi_{\Phi}(D) := -\frac{1}{2\pi i} \int_{\partial D} (Q_{\theta,\Phi}^\OO - \zeta)^{-1} d\zeta,\quad m_\Phi(D):=\rank\Pi_\Phi(D),
\end{equation}
then $\Pi_{\Phi}(D)$ and $\Pi_\OO(D)$ have the same rank for any $\Phi\in\UU_\delta(\OO)$ if $\delta$ is sufficiently small. Since $m_\Phi(D) = m_{\Phi(\OO)}(D)$ by \eqref{eqn:QthetaPhiOO}, we conclude that
\begin{equation}
\label{mult stability}
    m_{\Phi(\OO)}(D) \textrm{ is constant for }\Phi\in\UU_\delta(\OO) \textrm{ if $\delta$ is sufficiently small}.
\end{equation}

We note that for every $z_0$ and $\OO$, one of the following cases has to occur:
\begin{equation}
\label{good case}
    \forall\,\delta>0,\quad\exists\,\Phi\in\UU_\delta(\OO)\ \textrm{ such that } m_{\Phi(\OO)}(z_0) < m_{\Phi(\OO)}(D),
    \end{equation}
or
\begin{equation}
\label{bad case}
    \exists\,\delta>0,\ \textrm{such that}\quad\forall\,\Phi\in\UU_\delta(\OO),\ m_{\Phi(\OO)}(z_0) = m_{\Phi(\OO)}(D).
\end{equation}
Assuming \eqref{good case} we can prove the lemma by induction on $m_\OO(z_0)$. If $m_\OO(z_0)=1$, \eqref{mult stability} shows that $m_{\Phi(\OO)}(D)=1$ for $\Phi\in\UU_\delta(\OO)$ with $\delta$ small. It then follows from \eqref{good case} that we can find $\Phi\in\UU_\delta(\OO)$ such that $m_{\Phi(\OO)}(z_0)<1$, i.e. $z_0\notin\Spec(Q_{\theta}^{\Phi(\OO)})$. Assuming that we proved the lemma in the case $m_\OO(z_0)<M$, we now assume that $m_\OO(z_0)=M$. We note that for any $\Phi_1\in\Diff(\OO)$ and $\Phi_2\in\Diff(\Phi_1(\OO))$,
\[    \|\Phi_2\circ\Phi_1 -\id\|_{C^2} \leq  C (\|\Phi_1 -\id\|_{C^2} + \|\Phi_2 -\id\|_{C^2} ),
\]
where $C$ is a constant depending only on the dimension $n$. For any $\delta>0$, \eqref{good case} implies that we can find $\Phi_1\in\Diff(\OO)$ with $\|\Phi_1-\id\|_{C^2}<\delta/2C$ such that $m_{\Phi_1(\OO)}(z_0)<M$. It then follows from our induction hypothesis that there exists $\Phi_2\in\Diff(\Phi_1(\OO))$ with $\|\Phi_2-\id\|_{C^2}< \delta/2C$ such that $z_0\notin \Spec(Q_\theta^{\Phi_2(\Phi_1(\OO))})$. We now take $\Phi=\Phi_2\circ\Phi_1$, then $\Phi\in \UU_\delta(\OO)$ and $z_0\notin \Spec(Q_\theta^{(\Phi(\OO)})$.

It remains to show that \eqref{bad case} is impossible. For that, we shall argue by contradiction, assume that $m_{\OO}(D)=M$ and that \eqref{bad case} holds. For $\Phi\in\UU_\delta(\OO)$, we define
\[  
    k(\Phi) := \min\{ k : (Q_{\theta,\Phi}^\OO - z_0)^k \Pi_\Phi(D) = 0 \},
\]
then $1\leq k(\Phi)\leq M$. It follows from \eqref{eqn:QthetaPhiOO} and \eqref{eqn:Proj QthetaPhiOO} that if $\|\Phi_j - \Phi\|_{C^{2M}}\to 0$ and  $(Q_{\theta,\Phi_j}^\OO - z_0)^k \Pi_{\Phi_j}(D) = 0$, then $(Q_{\theta,\Phi}^\OO - z_0)^k \Pi_\Phi(D) = 0$. We now put
\[ 
    k_0 := \max\{ k(\Phi) : \Phi\in \UU_{\delta/2}(\OO) \}, 
\]
and assume that the maximum is attained at $\Phi_0\in\UU_{\delta/2}(\OO)$ i.e. $k(\Phi_0)=k_0$, then there exists $\delta'>0$ such that $\|\Phi-\Phi_0\|_{C^{2M}} < \delta' \Rightarrow k(\Phi)=k_0$. Henceforth, we can replace our original obstacle $\OO$ with $\Phi_0(\OO)$, decrease $\delta$ and then assume by \eqref{bad case} that
\begin{equation}
\label{bad case 2}
    \begin{gathered}
    ( Q_{\theta,\Phi}^\OO - z_0 )^{k_0} \Pi_{\Phi}(D) = 0,\quad  ( Q_{\theta,\Phi}^\OO - z_0 )^{k_0 - 1} \Pi_{\Phi}(D) \neq 0, \\
    m_{\Phi}(z_0) = \rank \Pi_{\Phi}(D) = M,\quad\forall\,\Phi\in\Diff(\OO),\ \|\Phi-\id\|_{C^{2M}} < \delta. \end{gathered}
\end{equation}

Before proving that \eqref{bad case 2} is impossible we introduce a family of deformations in $\Diff(\OO)$ acting near a fixed point on $\partial\OO$. For any fixed $x_0\in\partial\OO$ and some $h_0>0$ small we can choose a family of functions $\chi_h\in \CI(\partial\OO;[0,\infty))$ depending continuously in $h\in(0,h_0]$ with
\begin{equation}
\label{eqn:chi_h}
    \int_{\partial\OO} \chi_h(x) dS(x) = 1,\quad \supp\chi_h\subset B_{\partial\OO}(x_0,h),\quad\forall\,h\in(0,h_0],
\end{equation}
where $B_{\partial\OO}(x_0,h)$ denotes the geodesic ball on $\partial\OO$ with center $x_0$ and radius $h$. For each $h\in (0,h_0]$, we construct a smooth vector field $V_h\in\CIc(\RR^n;\RR^n)$ with some small constant $\delta_h=\OO(h^{2M+n-1})$ such that
\begin{equation}
\label{eqn:Vh}
    \begin{gathered}
        V_h(x) = \delta_h \chi_h(x) \nu_g (x),\ \forall x\in \partial\OO,\quad \|V_h\|_{C^{2M}}<\epsilon/2, \\ 
        \supp V_h \subset B_{\RR^n}(x_0,Ch)\textrm{ for some }C>0,
    \end{gathered}    
\end{equation}
where $\nu(x)$ is the normal vector at $x\in\partial\OO$ pointing inward. Let $\varphi_h^t:\RR^n\to\RR^n$ be the flow generated by the vector field $V_h$. It follows from \eqref{eqn:Vh} that for every $h\in(0,h_0]$ there exists $t_0>0$ such that
\[    
    \varphi_h^t\in\Diff(\OO),\quad\|\varphi_h^t-\id\|_{C^{2M}} < \delta,\quad\forall\,t\in(-t_0,t_0).
\]

Assuming \eqref{bad case 2} we can find $w\in L^2(\Gamma_\theta\setminus\OO)$ so that 
$u := (Q_{\theta}^\OO - z_0)^{k_0 -1} \Pi_\OO (D)w \neq 0$. For any fixed $x_0\in\partial\OO$ and $h\in(0,h_0]$, we take $\Phi_t:=\varphi_h^{t}$, $t\in(-t_0,t_0)$ and put   
\[
        u(t) := ({\Phi_t}^{-1})^* v(t),\quad  v(t) := (Q_{\theta,\Phi_t}^\OO-z_0)^{k_0 -1} \Pi_{\Phi_t}(D) w.
\] 
In view of \eqref{eqn:QthetaPhiOO}, $(Q_{\theta,\Phi_t}^\OO-z_0) v(t) = 0$ implies that
\begin{equation}
\label{eqn:outgoing u_t}
    (Q_\theta - z_0)u(t) = 0\quad\textrm{in }\Gamma_\theta\setminus\Phi_t(\OO).
\end{equation}
Since $\Phi_t(\OO)\subset\OO$ for $t\geq 0$, we can restrict \eqref{eqn:outgoing u_t} to the region $\Gamma_\theta\setminus\OO$ then differentiate it in $t$, by taking $t=0$, we obtain that
\begin{equation}
\label{eqn:u'(0)}
    (Q_\theta - z_0) u'(0) = 0 \quad\textrm{in }\Gamma_\theta\setminus\OO.
\end{equation}
Recalling that $u(t,x)=v(t,\varphi_h^{-t} x)$ and $u(0)=v(0)=u$, we conclude from the flow equation that $u'(0)=v'(0)-\partial_x u \cdot V_h$, thus by \eqref{eqn:Vh} we have
\begin{equation}
\label{eqn:u'(0) on boundary}
    u'(0) = -\delta_h \chi_h(x) \partial_{\nu_g} u,\quad\textrm{on }\partial\OO.
\end{equation}
We now multiply \eqref{eqn:u'(0)} by $u$ then integrate it on $\Gamma_\theta\setminus\OO$, then
\begin{equation}
\label{eqn:IBP Qtheta}
    \begin{split}
        0 & = \int_{\Gamma_\theta\setminus\OO} u\, (Q_\theta -z_0)u'(0) \\
        & = \int_{\Gamma_\theta\setminus\OO} u'(0)\, (Q_\theta -z_0)u  + \int_{\Gamma_\theta\setminus\OO} \sum_{j,k}\partial_j(u'(0)g^{jk}\partial_k u - u g^{jk}\partial_k u'(0) ) \\
        & = \int_{\partial\OO} (u'(0)\,\partial_{\nu_g} u - u\,\partial_{\nu_g} u'(0))\, dS.
    \end{split}
\end{equation}
It then follows from $u|_{\partial\OO}=0$ and \eqref{eqn:u'(0) on boundary} that 
\[  
0=\int_{\partial\OO}\chi_h(x) (\partial_{\nu_g} u(x))^2 dS(x) ,
\]
sending $h\to 0+$, we conclude from \eqref{eqn:chi_h} that $\partial_{\nu_g}u(x_0) = 0$. We note that $x_0\in\partial\OO$ can be chosen arbitrarily, thus $\partial_{\nu_g} u|_{\partial\OO} \equiv 0$. Putting $\Tilde{u}:= 1_{\OO}\cdot 0 + 1_{\Gamma_\theta\setminus\OO}\cdot u$, the same arguments as in the proof of Lemma \ref{lem:Nepstheta Fredholm} show that $\Tilde{u}\in H^2(\Gamma_\theta)$ and $(Q_\theta - z_0)\Tilde{u}=0$ on $\Gamma_\theta$. But unique continuation results
for second order elliptic differential equations show that $\Tilde{u}\equiv 0$, thus a contradiction.
\end{proof}

Now we consider the behavior of $\Spec(P^\OO)$ under the deformations of $\OO$. In the notation of \S \ref{reference operator}, for $\Phi\in\Diff(\OO)$, the pullback $\Phi^*$ gives an isomorphism between $\HH^{\Phi(\OO)}$ and $\HH^\OO$, which also restricts to an isomorphism between $\DD^{\Phi(\OO)}$ and $\DD^\OO$. Like \eqref{eqn:QthetaPhiOO} we define the deformed operator of $P^\OO$ associate with $\Phi$:
\begin{equation}
\label{eqn:PPhiOO}
    P_\Phi^\OO := \Phi^* P^{\Phi(\OO)} (\Phi^*)^{-1},\quad\textrm{ with domain } \DD^\OO.
\end{equation}
Since $(P^{\Phi(\OO)} + i)^{-1}$ is compact by Lemma \ref{lem:reference operator}, the same holds for $P_\Phi^\OO$, it follows that $P_\Phi^\OO$ has a discrete spectrum. Moreover, $\Spec(P_\Phi^\OO)$ must be identical to $\Spec(P^{\Phi(\OO)})$, which lies in $\RR$ due to the self-adjointness of $P^{\Phi(\OO)}$. 

Before stating the deformation results for $\Spec(P^\OO)$, we notice that unlike Lemma \ref{lem:deform obstacle}, there is a subset of $\Spec(P^\OO)$ which is invariant under the deformations of the obstacle, that is the compactly supported embedded eigenvalues of $P$,
\begin{equation}
\label{defn:SpecComp of P}
\Spec_{\comp} (P) := \{\lambda\in\CC : \exists\,0\neq u\in\DD_{\comp}\ \textrm{such that }(P-\lambda)u=0 \},
\end{equation}
where $\DD_{\comp} := \{ u\in\DD : u|_{\RR^n\setminus B(0,R_0)}\in H^2_{\comp}(\RR^n\setminus B(0,R_0)) \}$.
In view of the unique continuation results
for second order elliptic differential equations, $u$ in \eqref{defn:SpecComp of P} must vanish on $\RR^n\setminus B(0,R_0)$, thus $u\in \DD^\sharp$ for any $\OO$ containing $\overline{B(0,R_0)}$, which implies that $\Spec_{\comp} (P) \subset \Spec(P^\OO)$. The next lemma shows that any eigenvalue of $P^\OO$ other than those compactly supported embedded eigenvalues of $P$ can still be moved by deforming the obstacle:
\begin{lem} 
\label{lem:deform obstacle P sharp}
Suppose that the obstacle $\OO\subset B(0,R_1)$ contains $\overline{B(0,R_0)}$ and $z_0\in\Spec(P^\OO)\setminus\Spec_{\comp}(P)$, then for any $\delta>0$ there exists $\Phi\in\Diff(\OO)$ with $\|\Phi-\id\|_{C^2} < \delta$ such that $z_0 \notin \Spec(P^{\Phi(\OO)})$.
\end{lem}

\begin{proof}
The proof is similar to Lemma \ref{lem:deform obstacle} except that we need a different approach from \eqref{eqn:IBP Qtheta} since the integration by parts is not available in the black box. Suppose that $z_0\in\Spec(P^\OO)$ with multiplicity $m_\OO^P (z_0)$ and that $P^\OO$ has exactly one eigenvalue in $D(z_0,2r)$. For $D:=D(z_0,r)$ we put
\[ 
    \Pi_\OO^P (D) := -\frac{1}{2\pi i} \int_{\partial D} (P^\OO - \zeta)^{-1} d\zeta,\quad m_\OO^P (D):=\rank\Pi_\OO^P (D).
\]
Using \eqref{eqn:P sharp} and \eqref{eqn:A Phi} we can deduce that $\partial D\owns \zeta \mapsto (P_\Phi^\OO - \zeta)^{-1}$ exists for $\Phi\in\UU_\delta(\OO)$ with $\delta$ small enough, then we define
\[
    \Pi_\Phi^P (D) := -\frac{1}{2\pi i} \int_{\partial D} (P_\Phi^\OO - \zeta)^{-1} d\zeta,\quad m_\Phi^P (D):=\rank\Pi_\Phi^P (D) = m_{\Phi(\OO)}^P (D).
\]
We remark that $m_\OO^P(D)$ is also invariant under small deformations of obstacles:
\begin{equation}
\label{mult stability 2}
    m_{\Phi(\OO)}^P (D) \textrm{ is constant for }\Phi\in\UU_\delta(\OO) \textrm{ if $\delta$ is sufficiently small}.
\end{equation}
In view of the proof of Lemma \ref{lem:deform obstacle}, it is enough to exclude the following case:
\begin{equation}
\label{bad case 3}
    \exists\,\delta>0,\ \textrm{such that}\quad\forall\,\Phi\in\UU_\delta(\OO),\ m_{\Phi(\OO)}^P (z_0) = m_{\Phi(\OO)}^P (D).
\end{equation}
Again we argue by contradiction, assume that \eqref{bad case 3} holds and $m_\OO^P(D) = M \geq 1$. We remark that unlike the proof of Lemma \ref{lem:deform obstacle}, the self-adjointness of $P^{\Phi(\OO)}$ implies that $(P^{\Phi(\OO)} - z_0) \Pi_{\Phi(\OO)}^P (D) = 0$ thus $(P_\Phi^\OO - z_0)\Pi_\Phi^P (D) = 0$ for any $\Phi\in\UU_\delta(\OO)$. We now choose $w\in\HH^\OO$ such that $u:=\Pi_{\OO}^P (D) w\neq 0$. For any fixed $x_0\in \partial\OO$ and $h\in (0,h_0]$, we set $\Phi_t :=\varphi_h^t$ where $\varphi_h^t$ is the flow generated by $V_h$ given in \eqref{eqn:Vh}, there exists $t_0>0$ such that $\Phi_t\in\UU_\delta(\OO)$ for all $-t_0<t<t_0$. Let
\[
    v(t) := \Pi_{\Phi_t}^P (D) w\in\DD^\OO, \quad u(t):=(\Phi_t^{-1})^* v(t), 
\]
we have $(P_{\Phi_t}^\OO - z_0)v(t) = 0$, thus $(P^{\Phi_t(\OO)} - z_0) u(t) = 0$. Recalling \eqref{eqn:P sharp} we obtain that for some $\psi\in\CIc(\OO)$, $\psi=1$ near $\overline{B(0,R_0)}$ and $t_0>0$ small enough,
\begin{equation}
\label{eqn:u(t) in obstacle}
    \forall\,t\in(-t_0,t_0),\quad P(\psi u(t)) + Q ((1-\psi)u(t)) - z_0 u(t) = 0\quad \textrm{in } \Phi_t(\OO).
\end{equation}
Since $\Phi_t(\OO)\supset \OO$ for $t\leq 0$, we can restrict \eqref{eqn:u(t) in obstacle} to $\OO$ and differentiate it in $t$, by taking $t=0$, we have
\begin{equation}
\label{eqn:u'(0) in obstacle}
    P(\psi u'(0)) + Q((1-\psi)u'(0)) - z_0 u'(0) = 0\quad\textrm{in }\OO.
\end{equation}
Next we compute the inner product of the left hand side and $u$ on the Hilbert space $\HH^\OO$ defined by \eqref{eqn:H sharp}. For that, choose $\psi_j\in\CIc(\OO)$, $\psi_j=1$ near $\overline{B(0,R_0)}$, so that  
\begin{equation}
\label{eqn:psi1 psi2}
    \psi_1 = 1 \textrm{ near }\supp\psi,\quad \psi = 1 \textrm{ near }\supp\psi_2.
\end{equation}
Then we have, using the self-adjointness of $P$,
\[
    \langle P(\psi u'(0)) , u\rangle_{\HH^\OO} = \langle P(\psi u'(0)) , \psi_1 u \rangle_{\HH} = \langle \psi u'(0) , P(\psi_1 u) \rangle_{\HH},
\]
and $\langle Q((1-\psi)u'(0)) , u \rangle_{\HH^\OO} = \langle Q((1-\psi)u'(0)) , (1-\psi_2) u \rangle_{L^2(\OO)}$. Recalling \eqref{eqn:u'(0) on boundary}, integration by parts as in \eqref{eqn:IBP Qtheta} shows that
\[
\begin{split}
    &{\ \quad} \langle Q((1-\psi)u'(0)) , (1-\psi_2) u \rangle_{L^2(\OO)} - \langle (1-\psi)u'(0) , Q((1-\psi_2) u) \rangle_{L^2(\OO)} \\
    & = \int_\OO \sum_{j,k} \partial_j ((1-\psi)u'(0)g^{jk}\partial_k ((1-\psi_2)\Bar{u}) - (1-\psi_2)\Bar{u} g^{jk}\partial_k ((1-\psi)u'(0)))\\
    & = \int_{\partial\OO} -u'(0)\partial_{\nu_g}\Bar{u} + \Bar{u}\partial_{\nu_g}u'(0) = \int_{\partial\OO} \delta_h \chi_h |\partial_{\nu_g} u|^2 .
\end{split}
\]
It follows from \eqref{eqn:P sharp} and \eqref{eqn:psi1 psi2} that
\[
    \langle \psi u'(0) , P(\psi_1 u) \rangle_{\HH} = \langle u'(0) , \psi (P^\OO u - Q((1-\psi_1)u)) \rangle_{\HH^\OO} = \langle u'(0) , \psi P^\OO u \rangle_{\HH^\OO};
\]
and that
\[
\begin{split}
    \langle (1-\psi)u'(0) , Q((1-\psi_2) u) \rangle_{L^2(\OO)}
    & = \langle u'(0) , (1-\psi)(P^\OO u - P(\psi_2 u)) \rangle_{\HH^\OO} \\ 
    & = \langle u'(0) , (1 -\psi) P^\OO u \rangle_{\HH^\OO}.
\end{split}
\]
We now conclude from \eqref{eqn:u'(0) in obstacle} and all the calculation above that
\[
    0 = \langle u'(0) , (P^\OO - z_0) u \rangle_{\HH^\OO} + \int_{\partial\OO} \delta_h \chi_h |\partial_{\nu_g} u|^2 = \int_{\partial\OO} \delta_h \chi_h |\partial_{\nu_g} u|^2 .
\]
It follows that $\partial_{\nu_g} u (x_0) = 0$. Since $x_0\in\partial\OO$ can be chosen arbitrarily, we obtain that $\partial_{\nu_g} u|_{\partial\OO} \equiv 0$. Putting $\Tilde{u}:= 1_{\OO}u + 1_{\RR^n\setminus\OO}\cdot 0$, the same arguments as in the proof of Lemma \ref{lem:Nepstheta Fredholm} show that $\Tilde{u}\in \DD$ and $(P - z_0)\Tilde{u}=0$, which would imply that $z_0\in\Spec_{\comp}(P)$, a contradiction.
\end{proof}

\section{Proof of convergence}
\label{section:poc}

Before proving the convergence of eigenvalues of $P_\epsilon$ to resonances as $\epsilon\to 0+$, we recall a basic estimate of decay of the Green function of $Q_\theta^\OO$ off the diagonal $\{ (x,x): x\in\Gamma_\theta\setminus\OO \}$. For a detailed account see Shubin \cite{shubin} and references given there.  
\begin{lem}
\label{lem:green function}
Suppose that the obstacle $\OO\subset B(0,R_1)$ contains $\overline{B(0,R_0)}$ and that $z_0\notin\Spec(Q_\theta^\OO)$ with $-2\theta < \arg z_0 < 3\pi/2 + 2\theta$. The Schwartz kernel of the resolvent $(Q_\theta^\OO - z_0)^{-1} : L^2(\Gamma_\theta\setminus\OO)\to L^2(\Gamma_\theta\setminus\OO)$ is denoted by $G(z_0;x_\theta,y_\theta)$, where $x_\theta=f_\theta(x)$ is the parametrization on $\Gamma_\theta$. Then there exists $\beta>0$ such that for every $\delta>0$ there exists $C_\delta>0$ such that
\[
    |G(z_0;f_\theta(x),f_\theta(y))| \leq C_\delta\, e^{-\beta|x-y|}\quad\textrm{if}\quad |x-y|>\delta.
\]
\end{lem}

\begin{proof}
Identifying $\Gamma_\theta$ and $\RR^n$ by means of $f_\theta$, the pullback $f_\theta^*$ gives an isomorphism between $L^2(\Gamma_\theta\setminus\OO)$ and $L^2(\RR^n\setminus\OO)$ since there exists $C>0$ such that
\[ 
    C^{-1} < |\Det df_\theta(x)| = |x|^{1-n}|g_\theta(|x|)|^{n-1} |g_\theta '(|x|)| < C,\quad\textrm{for all }x. 
\]
Let $\Tilde{Q}_\theta^\OO := f_\theta^* Q_\theta^\OO (f_\theta^*)^{-1} : L^2(\RR^n\setminus\OO)\to L^2(\RR^n\setminus\OO)$ then $\Tilde{Q}_\theta^\OO$ is elliptic and equipped with the domain $H^2(\RR^n\setminus\OO)\cap H_0^1(\RR^n\setminus\OO)$. Moreover, $(\Tilde{Q}_\theta^\OO - z_0)^{-1}$ exists and we denote its Schwartz kernel by $\Tilde{G}(z_0;x,y)$, $x,y\in\RR^n\setminus\OO$, i.e.  $\Tilde{G}(z_0;x,y)=[(\Tilde{Q}_\theta^\OO - z_0)^{-1}\delta_y(\cdot)](x)$ where $\delta_y$ is the Dirac function supported at $y$. 

The same arguments as in \cite[Appendix 1]{shubin} show that there exists $\beta>0$ such that for every $\delta>0$ there exists $C_\delta>0$ such that
\[
    |\Tilde{G}(z_0;x,y)| \leq C_\delta\, e^{-\beta|x-y|}\quad\textrm{if}\quad |x-y|>\delta.
\]
We remark that the assumption in \cite[Appendix 1.1]{shubin} that the manifold $M$ is complete can be dropped if we introduce $\Tilde{d}(x,y)$, the substitute with smoothness properties for the distance $|x-y|$, on the whole $\RR^n$ then restrict it to $\RR^n\setminus\OO$. The remaining arguments in \cite[Appendix 1.2]{shubin} is still valid if we replace $M$ by $\RR^n\setminus\OO$. 

Using $(\Tilde{Q}_\theta^\OO - z_0)^{-1} = f_\theta^* (Q_\theta^\OO - z_0)^{-1} (f_\theta^*)^{-1}$ we obtain that 
\[
    G(z_0;f_\theta(x),f_\theta(y)) = (\Det df_\theta(y))^{-1} \Tilde{G}(z_0;x,y),\quad x,y\in\RR^n\setminus\OO,
\]
the desired estimate of $G(z_0;x_\theta,y_\theta)$ then follows from the estimate of $\Tilde{G}(z_0;x,y)$.
\end{proof}

We now state a more precise version of Theorem \ref{t:1}:
\begin{thm}
Suppose that $\Omega\Subset \{ z :\,-2\theta_0 < \arg z < 3\pi/2 + 2\theta_0 \}$. Then exists $\delta_0=\delta_0(\Omega)>0$ such that $\forall\,0<\delta<\delta_0$, $\exists\,\epsilon_\delta>0$ such that 
\begin{equation}
\label{eqn:convergence 1}    
0<\epsilon<\epsilon_\delta \implies \Spec(P_\epsilon)\cap \Omega_\delta \subset \bigcup_{j=1}^J D(z_j,\delta), 
\end{equation}
where $z_1,\cdots,z_J$ are the resonances of $P$ in $\Omega$ and $\Omega_\delta := \{z\in\Omega : \dist(z,\partial\Omega)>\delta \}$. Furthermore, for each resonance $z_j$ with the multiplicity $m(z_j)$ given by \eqref{eqn:multiplicity z0},   
\begin{equation}
\label{eqn:convergence 2}
    \#\,\Spec(P_\epsilon) \cap D(z_j,\delta) = m(z_j),\quad\forall\,0<\epsilon<\epsilon_\delta,
\end{equation}
where the eigenvalue in $\Spec(P_\epsilon)$ is counted with multiplicity defined in \eqref{eqn:Pepstheta mult}.
\end{thm}

\begin{proof}
First we put $\delta_0=\frac{1}{2}\min_{1\leq j\leq J} \dist(z_j,\partial\Omega)$ and fix $\theta\in[0,\theta_0)$ such that $\Omega\Subset \{ z :\,-2\theta < \arg z < 3\pi/2 + 2\theta \}$. To prove \eqref{eqn:convergence 1} we argue by contradiction. Suppose that there exist some $\delta<\delta_0$ and a sequence $\epsilon_k\to 0 +$ such that
\[    \exists\, z_k\in\Spec(P_{\epsilon_k})\cap\Omega_\delta\setminus \bigcup_{j=1}^J D(z_j,\delta),\quad k=1,2,\cdots
\]
Then there exists a subsequence $z_{n_k}\to z_0$, as $k\to\infty$, for some $z_0\in\overline{\Omega_\delta}\setminus\bigcup_{j=1}^J D(z_j,\delta)$. Since $z_0\in\Omega$, we see that $z_0$ is not a resonance, thus $P_\theta - z_0$ is invertible by definition. We may assume that $D(z_0,r)$ is disjoint with $\Spec(P_\theta)$ for some $r>0$, it then follows from Lemma \ref{lem:Pepstheta resolvent norm} that $\Spec(P_{\epsilon,\theta})\cap D(z_0,r) = \emptyset$ for $\epsilon$ small enough. However, Lemma \ref{lem:Pepstheta mult} shows that $\Spec(P_{\epsilon_{n_k},\theta})= \Spec(P_{\epsilon_{n_k}})\owns z_{n_k} \to z_0$ while $\epsilon_{n_k}\to 0 +$, which gives a contradiction.

It remains to prove \eqref{eqn:convergence 2}. For each resonance $z_j$, let 
\[
    V_j:=\{ u\in\DD_{\comp} : (P-z_j)u=0 \},
\]
then $V_j$ is finite dimensional and $V_j\neq \{0\}$ if and only if $z_j\in\Spec_{\comp}(P)$. We remark that $V_j$ is a subspace of $\HH_{R_0}$ given in \eqref{eqn:Hilbert space}, as a consequence of the unique continuation results for second order elliptic equations. The self-adjointness of $P$ implies that $V_1 \perp \cdots\perp V_J$ in the Hilbert space $\HH$. Putting $V_0:=V_1\oplus\cdots\oplus V_J$, $\HH$ admits the following orthogonal decomposition:
\begin{equation}
\label{eqn:orthogonal decomposition tilde H}
    \HH = V_0\oplus \Tilde{\HH}_{R_0} \oplus L^2(\RR^n\setminus B(0,R_0)).
\end{equation}
Let $\Pi_0:\HH\to V_0$ be the orthogonal projection. Since $V_0$ is an invariant subspace under $P$, we can introduce the restriction of $P$ as follows: 
\[\Tilde{P}:\Tilde{\HH}_{R_0} \oplus L^2(\RR^n\setminus B(0,R_0))\to \Tilde{\HH}_{R_0} \oplus L^2(\RR^n\setminus B(0,R_0)),\quad \Tilde{P}u:=(I-\Pi_0)Pu.\]
If we replace $\HH_{R_0}$ with $\Tilde{\HH}_{R_0}$ and replace $P$ by $\Tilde{P}$, which is also self-adjoint with domain $\Tilde{\DD}:=(I-\Pi_0)\DD$, it is easy to see that the assumptions \eqref{eqn:P} -- \eqref{eqn:Q defn} are still satisfied. Recalling the definition of resonances introduced in \S \ref{complex scaling}, any resonance of $\Tilde{P}$ must also be a resonance of $P$ and we have
\[  
    m(z_j)=\rank\oint_{z_j} (z-\Tilde{P}_\theta)^{-1} dz + \dim V_j.
\]
Note that $V_j\neq\{0\}$ implies that $z_j\in\Spec(P_\epsilon)$ for every $\epsilon>0$. Putting $\Tilde{P}_\epsilon := \Tilde{P} - i\epsilon(1-\chi(x))x^2$, it follows that 
\[
    \#\,\Spec(P_\epsilon) \cap D(z_j,\delta) = \#\,\Spec(\Tilde{P}_\epsilon) \cap D(z_j,\delta) + \dim V_j,\quad\forall\,\epsilon>0,
\]
while both sides are counted with multiplicities. Hence it is enough to establish \eqref{eqn:convergence 2} for $\Tilde{P}$. In other words, it suffices to prove \eqref{eqn:convergence 2} in the case that $P$ has no compactly supported embedded eigenvalues in $\Omega$.

Now we assume that $\Spec_{\comp}(P)\cap\Omega=\emptyset$. Lemma \ref{lem:deform obstacle} and \ref{lem:deform obstacle P sharp} show that there exists an obstacle $\OO\subset B(0,R_1)$ containing $\overline{B(0,R_0)}$ such that $\chi$ in \eqref{eqn:Peps} is equal to $1$ near $\OO$ and that $z_j\notin\Spec(P^\OO)\cup\Spec(Q_\theta^\OO)$, $j=1,\cdots,J$. Then we can decrease $\delta_0$ such that $\Spec(P^\OO)$ and $\Spec(Q_\theta^\OO)$ are disjoint with $\bigcup_{j=1}^J D(z_j,2\delta_0)$. For each $\delta\in(0,\delta_0)$, we can also decrease $\epsilon_\delta$ in \eqref{eqn:convergence 1} such that 
\[
    \forall\,0\leq \epsilon< \epsilon_\delta,\quad \bigcup_{j=1}^J D(z_j,2\delta)\cap\Spec(Q_{\epsilon,\theta}^\OO) = \emptyset.
\]
This follows from Lemma \ref{lem:Pepstheta resolvent norm} applied with $P_\theta = Q_\theta^\OO$ and $\Omega=\bigcup_{j=1}^J D(z_j,2\delta)$. Hence the Dirichlet-to-Neumann operators $\widehat{\NN}_{\epsilon,\theta}(z),\ 0\leq\epsilon<\epsilon_\delta$ introduced in \S \ref{section:N operator}, are well-defined for $z\in\bigcup_{j=1}^J D(z_j,2\delta)$. In view of \eqref{eqn:convergence 1}, Lemma \ref{lem:Pepstheta mult} and \ref{lem:DtoN} we obtain that $\partial D(z_j,\delta)\owns w\mapsto \widehat{\NN}_{\epsilon,\theta}(w)^{-1}$ exists and that for all $0<\epsilon<\epsilon_\delta$, $j=1,\cdots,J$,
\begin{equation}
\label{eqn:count eigenvalues using NN}
    \#\,\Spec(P_\epsilon) \cap D(z_j,\delta) = \frac{1}{2\pi i}\tr \int_{\partial D(z_j,\delta)} \widehat{\NN}_{\epsilon,\theta}(w)^{-1} \partial_w \widehat{\NN}_{\epsilon,\theta}(w)dw.
\end{equation}
In order to apply the Gohberg--Sigal--Rouch\'e theorem, we need the estimate:
\begin{equation}
\label{eqn:Rouche estimate}
    \forall\,0<\epsilon<\epsilon_\delta,\quad
    \|\widehat{\NN}_{\epsilon,\theta}(w) - \widehat{\NN}_{\theta}(w)\|_{H^{3/2}(\partial\OO)\to H^{3/2}(\partial\OO)} < 1,\quad w\in\partial D(z_j,\delta),
\end{equation}
here we write $\widehat{\NN}_\theta(\cdot)=\widehat{\NN}_{0,\theta}(\cdot)$ for simplicity. To obtain this estimate, we first choose $E^\out$ in \eqref{eqn:N out with extension} such that $\chi=1$ near $\supp E^\out \varphi$ for any $\varphi\in H^{3/2}(\partial\OO)$, then \eqref{eqn:N out with extension} reduces to $\NN_{\epsilon,\theta}^\out (z) \varphi = \partial_{\nu_g} (E^\out \varphi - (Q_{\epsilon,\theta}^\OO - z)^{-1}(Q-z)E^\out \varphi)$. Therefore,
\[
    (\widehat{\NN}_{\epsilon,\theta}(w) - \widehat{\NN}_{\theta}(w))\varphi
     = \langle D_{\partial\OO} \rangle^{-1} \partial_{\nu_g}((Q_{\theta}^\OO - w)^{-1}-(Q_{\epsilon,\theta}^\OO - w)^{-1})(Q-w)E^\out \varphi.
\]
Choosing $\psi\in\CIc(\Gamma_\theta\setminus\OO)$ such that $\psi=1$ near $\supp E^\out \varphi$, $\forall\varphi\in H^{3/2}(\partial\OO)$ and that $\chi=1$ near $\supp\psi$, \eqref{eqn:Rouche estimate} then follows from the following estimate: for $w\in\partial D(z_j,\delta)$,
\begin{equation}
\label{eqn:Rouche estimate 1}
    ((Q_{\theta}^\OO - w)^{-1}-(Q_{\epsilon,\theta}^\OO - w)^{-1})\psi = O_\delta(\epsilon):L^2(\Gamma_\theta\setminus\OO)\to H^2(\Gamma_\theta\setminus\OO).
\end{equation}
To obtain \eqref{eqn:Rouche estimate 1}, we denote the Schwartz kernel of the operator $(1-\chi)x_\theta^2 (Q_{\epsilon,\theta}^\OO - w)^{-1} \psi$ by $K(w;x_\theta,y_\theta)$. In the notation of Lemma \ref{lem:green function}, we have
\[ 
    K(w;f_\theta(x),f_\theta(y)) = (1-\chi(x))f_\theta(x)^2 G(w;f_\theta(x),f_\theta(y))\psi(y). 
\]
It follows from Lemma \ref{lem:green function} that there exists $\beta_\delta>0$ such that for all $w\in \partial D(z_j,\delta)$, $j=1,\cdots,J$, $|K(w;f_\theta(x),f_\theta(y))|\leq C|x|^2 e^{-\beta_\delta|x-y|}\psi(y)$, thus
\[
    \sup_{x_\theta} \int_{\Gamma_\theta\setminus\OO} |K(w;x_\theta,y_\theta)||dy_\theta|\leq C_\delta,\quad \sup_{y_\theta} \int_{\Gamma_\theta\setminus\OO} |K(w;x_\theta,y_\theta)||dx_\theta|\leq C_\delta.
\]
The Schur test shows that $(1-\chi)x_\theta^2 (Q_{\epsilon,\theta}^\OO - w)^{-1} \psi = O_\delta(1):L^2(\Gamma_\theta\setminus\OO)\to L^2(\Gamma_\theta\setminus\OO)$. Hence we can write
\[
    ((Q_{\theta}^\OO - w)^{-1}-(Q_{\epsilon,\theta}^\OO - w)^{-1})\psi = -i\epsilon(Q_{\epsilon,\theta}^\OO - w)^{-1} (1-\chi)x_\theta^2 (Q_{\theta}^\OO - w)^{-1}\psi.
\]
It remains to show that for $\epsilon_\delta>0$ small enough,
\[
    (Q_{\epsilon,\theta}^\OO - w)^{-1} = O_\delta(1): L^2(\Gamma_\theta\setminus\OO)\to H^2(\Gamma_\theta\setminus\OO),\quad w\in\bigcup_{j=1}^J \partial D(z_j,\delta),\ 0<\epsilon<\epsilon_\delta.
\]
This follows from Lemma \ref{lem:Pepstheta resolvent norm} with $P_\theta=Q_\theta^\OO$ and $\Omega=\bigcup_{j=1}^J \partial D(z_j,\delta)$. Using \eqref{eqn:Rouche estimate 1} we can decrease $\epsilon_\delta$ such that \eqref{eqn:Rouche estimate} holds for $j=1,\cdots,J$. Now we apply the Gohberg--Sigal--Rouch\'e theorem to conclude that for all $0<\epsilon<\epsilon_\delta$ and $j=1,\cdots,J$,
\[
    \frac{1}{2\pi i}\tr \int_{\partial D(z_j,\delta)} \widehat{\NN}_{\epsilon,\theta}(w)^{-1} \partial_w \widehat{\NN}_{\epsilon,\theta}(w)dw = \frac{1}{2\pi i}\tr \int_{\partial D(z_j,\delta)} \widehat{\NN}_{\theta}(w)^{-1} \partial_w \widehat{\NN}_{\theta}(w)dw.
\]
Finally, using Lemma \ref{lem:DtoN}, \eqref{eqn:count eigenvalues using NN} and the equation above, we obtain \eqref{eqn:convergence 2}.
\end{proof}

\def\arXiv#1{\href{http://arxiv.org/abs/#1}{arXiv:#1}}

\end{document}